\documentclass[journal,draftcls,onecolumn,12pt,twoside]{IEEEtranTCOM}
\usepackage{indentfirst}
\usepackage{color}
\usepackage{amsfonts,amsmath,amssymb,amsthm}
\usepackage{graphicx,multirow,bm}
\usepackage{cite}
\usepackage{arydshln}

\newtheorem{Theorem}{Theorem}
\newtheorem{Example}{Example}
\newtheorem{Remark}{Remark}
\newtheorem{Lemma}{Lemma}
\newtheorem{Construction}{Construction}

\newenvironment{psmallmatrix}
  {\left(\begin{smallmatrix}}
  {\end{smallmatrix}\right)}
\begin{document}
	\title{MDS Array Codes With (Near) Optimal Repair Bandwidth for All Admissible Repair Degree
	}
	
	\author{Jie~Li,~\IEEEmembership{Member,~IEEE,} Yi~Liu,~\IEEEmembership{Member,~IEEE,}
	        Xiaohu~Tang,~\IEEEmembership{Senior~Member,~IEEE,}
	        Yunghsiang S.~Han,~\IEEEmembership{Fellow,~IEEE,} Bo~Bai,~\IEEEmembership{Senior~Member,~IEEE,} and Gong Zhang
	\thanks{J. Li, B. Bai, and G. Zhang are with the Theory Lab, Central Research Institute, 2012 Labs, Huawei Technologies Co., Ltd., Shatin, New Territories, Hong Kong SAR, China (e-mails: li.jie9@huawei.com; baibo8@huawei.com; nicholas.zhang@huawei.com).}
	\thanks{Y. Liu was with the Information Coding and Transmission Key Lab of Sichuan Province, CSNMT Int. Coop. Res. Centre (MoST), Southwest Jiaotong University, Chengdu, 610031, China (e-mail: yiliu.swjtu@outlook.com).}
	\thanks{X. Tang is with the Information Coding and Transmission Key Lab of Sichuan Province, CSNMT Int. Coop. Res. Centre (MoST), Southwest Jiaotong University, Chengdu, 610031, China (e-mail: xhutang@swjtu.edu.cn).}
	\thanks{Yunghsiang S.~Han is with the Shenzhen Institute for Advanced Study, University of Electronic Science and Technology of China, Shenzhen, 518110, China (e-mail: yunghsiangh@gmail.com).}}

	\maketitle
	
	\begin{abstract}
	Abundant high-rate $(n, k)$ minimum storage regenerating (MSR) codes have been reported in the literature. However, most of them require contacting all the surviving nodes during a node repair process, resulting in a repair degree of $d=n-1$.  In practical systems, it may not always be feasible to connect and download data from all surviving nodes, as some nodes may be unavailable. Therefore, there is a need for MSR code constructions with a repair degree of $d<n-1$. Up to now, only a few $(n, k)$ MSR code constructions with repair degree $d<n-1$ have been reported, some have a large sub-packetization level, a large finite field, or restrictions on the repair degree $d$. In this paper, we propose a new $(n, k)$ MSR code construction that works for any repair degree $d>k$, and has a smaller sub-packetization level or finite field than {some existing constructions}. Additionally, in conjunction with a previous generic transformation to reduce the sub-packetization level, we obtain an MDS array code with a small sub-packetization level and $(1+\epsilon)$-optimal repair bandwidth (i.e., $(1+\epsilon)$ times the optimal repair bandwidth) for repair degree $d=n-1$. This code outperforms {some existing ones} in terms of either the sub-packetization level or the field size.	
	\end{abstract}


\section{Introduction}
In distributed storage systems, data are stored across multiple unreliable storage nodes. Thus, redundancy needs to be introduced to provide fault tolerance. Classic Maximum Distance Separable (MDS) codes can provide an optimal tradeoff between fault tolerance and storage overhead and thus is an efficient redundancy mechanism deployed for many years. However, repairing a failed node requires an excessive \textit{repair bandwidth}, defined as the amount of data downloaded to repair a failed node. 

One way to reduce the repair bandwidth is to use MDS array codes, where the codeword is an array of size $N\times n$ instead of a vector. For a distributed storage system encoded by an $(n,k)$ MDS array code, each node stores $N$ symbols, where $N$ is called the \textit{sub-packetization level}. {The cut-set  bound in \cite{dimakis2010network} shows that the repair bandwidth of $(n,k)$ MDS array codes with sub-packetization level $N$ is lower bounded by 
$\gamma_{\rm optimal}=\frac{d}{d-k+1}N$. Here, $d$ such that $k\le d<n$ denotes the number of helper nodes contacted during the repair process and is named \textit{repair degree}.} MDS array codes with repair bandwidth attaining this lower bound are said to have the \textit{optimal repair bandwidth} and are also referred to as MSR codes in \cite{dimakis2010network}.

During the past decade, various MSR codes have been proposed \cite{Goparaju,invariant subspace,hadamard,tian2014characterizing,Hadamard strategy,Long IT,Barg2,Zigzag,Sasidharan-Kumar2,Rashmi2011optimal,han2015update,li2017generic,li2018generic,hou2020binary,elyasi2020cascade,li2022generic,Barg1,chen2019explicit,liu2022generic,vajha2021small,li2022constructing,li2023msr,liu2022optimal}. However, in the high-rate (e.g., $\frac{k}{n}>\frac{1}{2}$) regime, existing constructions have two imperfections: i) 
most constructions have a repair degree of $d=n-1$, meaning that repairing a failed node requires contacting all the remaining surviving nodes. However, it is not always feasible to connect and download data from all the surviving nodes in a practical system, as some nodes may be unavailable due to other assigned jobs or network congestion \cite{mahdaviani2018product}; ii)  all the known $(n,k)$ MSR code constructions with repair degree $d=n-1$ require a significantly large sub-packetization level $N$, i.e., $N\ge r^{\frac{n}{r+1}}$, where $r=n-k$. This can lead to reduced design space in various system parameters and make managing meta-data difficult, hindering implementation in practical systems \cite{Rawat}. 

\subsection{Related work on $(n, k)$ MSR codes with repair degree $d<n-1$}

Up to now, only a few results on MSR codes with repair degree $d<n-1$ have been reported in the literature. In \cite{rawat2016progress,goparaju2017minimum}, the authors showed the existence of MSR codes with repair degree $d<n-1$, and some explicit constructions were given in  \cite{Barg1,liu2022generic,vajha2021small,chen2019explicit,li2022constructing,liu2022optimal}. In this paper, we  focus only on explicit constructions.

In Sections IV and VIII of \cite{Barg1}, Ye and Barg proposed two $(n, k)$ MSR codes with sub-packetization level $(d-k+1)^n$ by using diagonal matrices and permutation
matrices as the building blocks of the parity-check matrices.  In \cite{vajha2021small}, an $(n, k)$ MSR code with a smaller sub-packetization level of $(d-k+1)^{\frac{n}{d-k+1}}$ was generated, however,  $d$ is restricted to be $k+1, k+2, k+3$. 
In \cite{chen2019explicit}, Chen and Barg presented an $(n, k)$ MSR codes with a sub-packetization level of $(d-k+1)^n$.
In \cite{liu2022generic}, Liu \textit{et al.} gave an $(n, k)$ MSR code with a sub-packetization level of 
$(d-k+1)^{\frac{n}{2}}$.  Recently, in
\cite{li2022constructing}, an $(n, k)$ MSR code was constructed with a sub-packetization level of $2^{\frac{n}{3}}$ and a repair degree of $d=k+1$. This MSR code was generalized to support any repair degree $d$ with $d\in [k+1: n-1)$ and a sub-packetization level of $w^{\frac{n}{w+1}}$ in a follow-up work \cite{li2023msr}, where $w=d-k+1$, but requires searching over a finite field with a size larger than $nw+\sum\limits_{t=1}^{w+1}{w+1\choose t} \frac{wt(t-1)}{2}$. Despite the additional effort required in searching (i.e., explicit constructions are unknown for general code parameters $n, k$, and $d$), the MSR codes in \cite{li2022constructing} and  \cite{li2023msr} have the smallest sub-packetization level among all existing MSR codes with the same $n,k,d$. Independent and parallel to this work, Zhang and Zhou proposed an $(n, k)$ MSR code with a sub-packetization level of $2^{\frac{n}{2}}$ in a very recent work \cite{zhang2023vertical}, which is similar to the one proposed in this paper but requires a larger finite field when $d-k+1>2$.
For convenience, in this paper, these eight codes are referred to as YB code 1, YB code 2, VBK code, CB code, LLT code, WLHY code, LWHY code, and ZZ code, respectively. 

Overall, most of the aforementioned $(n, k)$ MSR codes with
repair degree $d<n-1$ either {have a large sub-packetization level (i.e., $N=(d-k+1)^n$)} \cite{Barg1,chen2019explicit} or are limited to only a few values of the repair degree $d$ \cite{vajha2021small,li2022constructing}. We want to point out that there are a few MSR codes with multiple repair degrees, e.g., \cite{liu2022optimal} and MSR codes in Sections V and IX of \cite{Barg1}, which are outside of the scope of this paper as we only focus on MSR codes with a single repair degree.

\subsection{Related work on $(n, k)$ MDS array codes with small sub-packetization level}
Large sub-packetization levels in codes can hinder their implementation in practical systems \cite{Rawat}, making it desirable to construct codes with small sub-packetization levels. Recent works have demonstrated that high-rate MDS array codes with small sub-packetization levels can be constructed by sacrificing the optimality of the repair bandwidth. 

In \cite{Rawat}, two high-rate MDS array codes with small sub-packetization levels and $(1+\epsilon)$-optimal repair bandwidth were proposed. The first code has
a sub-packetization level of $N=r^\tau$ and the repair bandwidth is no larger than $(1+\frac{1}{\tau})$ times the optimal repair bandwidth, where $\tau$ is an integer and 
$1\le \tau<\lceil\frac{n}{r}\rceil$. However,
this code is constructed over a significantly large finite field $\mathbf{F}_q$, i.e., $q\ge n^{(r-1)N+1}$, which may hinder its deployment
in practical systems. The second MDS array code is obtained by
combining an MDS array code with optimal repair bandwidth
and another error-correcting code with specific parameters.
For convenience,
we refer to these two codes as RTGE code 1 and
RTGE code 2 in this paper. 

Recently, a generic transformation was presented in \cite{eMSR_d_eq_n-1} that can convert any MSR code into an MDS array code with a small sub-packetization level and $(1+\epsilon)$-optimal repair bandwidth, resulting in several explicit MDS array codes with small sub-packetization levels. Note that all the MDS array codes in \cite{Rawat,eMSR_d_eq_n-1} have a repair degree of $d=n-1$.

\subsection{Main Contribution}

The main contribution of this paper is the derivation of a new $(n, k)$ MSR code construction with any given repair degree $d$ such that $k<d\le n-1$, {where its sub-packetization level is $w^{\lceil\frac{n}{2}\rceil}$ with $w=d-k+1$. The required field size $q$ is a prime power such that $q>\lceil\frac{n}{2}\rceil(w+2)$ if $w=2$, $q>\lceil\frac{n}{2}\rceil(w+1)$ if $2<w<r$,  and $q>\lceil\frac{n}{2}\rceil w$ if $w=r$}. {Compared with existing constructions, the new MSR code $\mathcal{C}_1$ has advantages in terms of either the sub-packetization level or the field size. Please refer to Tables \ref{Table comp MSR} and \ref{Table comp_duplication} for more details.
%
%
%
%

}

Furthermore, the MSR code $\mathcal{C}_1$ can also be used for $d=n-1$. When combined with the generic transformation in \cite{eMSR_d_eq_n-1}, we obtain a new $(n,k)$ MDS array code $\mathcal{C}_2$ {with a small sub-packetization level of $r^{\lceil\frac{n}{2s}\rceil}$ and repair degree $d=n-1$, where $r=n-k$, $s$ is any factor of $n$, and the require field size is $q>sr\lceil\frac{n}{2s}\rceil$}. {The sub-packetization level or finite field size of $\mathcal{C}_2$ is smaller than that of existing ones.} 

The remainder of the paper is organized as follows.
Section II reviews some necessary preliminaries of high-rate MDS array codes. The new $(n,k)$ MSR code $\mathcal{C}_1$ is presented in Section \ref{sec:C3}. Section \ref{sec:eMSRd=n-1} gives an MDS array code $\mathcal{C}_2$ with a small sub-packetization level. Section \ref{sec:comp} compares key parameters
among the MDS array codes proposed in this paper and some
existing ones. Finally, Section \ref{sec:conclusion} concludes 
the work.
\section{Preliminaries}
In this section, we introduce some preliminaries on MDS array codes and a special partition for a given basis. Throughout this paper, we assume that $q$ is a prime power and $\mathbf{F}_q$ is the finite field with $q$ elements. Let $[a: b)$ be the set $\{a, a+1, \ldots, b-1\}$ for two integers $a$ and $b$. For a matrix $A$, denote by {$A[a,b]$,} the $(a,b)$-th entry and $A[a,:]$ the $a$-th row, where $a,b\ge0$.

\subsection{$(n,k)$ Array Codes}
Let $\mathbf{f}_0, \mathbf{f}_1, \ldots, \mathbf{f}_{n-1}$ be the data stored across a distributed storage system consisting of $n$ nodes based on an $(n,k)$ array code, where $\mathbf{f}_i$ is a column vector of length $N$ over $\mathbf{F}_q$.  We consider   $(n,k)$ array codes defined by the following parity-check form:
\begin{equation}\label{Eqn parity check eq}
  \underbrace{\left(
        \begin{array}{cccc}
                A_{0,0} & A_{0,1} & \cdots & A_{0,n-1} \\
                A_{1,0} & A_{1,1} & \cdots & A_{1,n-1} \\
       \vdots & \vdots & \ddots & \vdots \\
            A_{r-1,0} & A_{r-1,1} & \cdots & A_{r-1,n-1} \\
            \end{array}
            \right)}_{A}
\left(\begin{array}{c}
              \mathbf{f}_0\\
              \mathbf{f}_1\\
             \vdots\\
         \mathbf{f}_{n-1}
            \end{array}
         \right)=\mathbf{0}_{rN},
\end{equation}
where $r=n-k$, $\mathbf{0}_{rN}$ denotes the zero column vector of length $rN$, and will be abbreviated as $\mathbf{0}$ in the sequel if its length is clear. The $rN\times nN$ block matrix $A$ in \eqref{Eqn parity check eq} is called the \textit{parity-check matrix} of the code, which can be written as
$
A=(A_{t,i})_{t\in [0: r),i\in[0: n)},$
to indicate the block entries, where $A_{t,i}$ is an $N\times N$ matrix.

Note that for each $t\in[0:r)$, $\sum\limits_{i=0}^{n-1}A_{t,i}\mathbf{f}_i=0$ contains $N$ {equations. For convenience}, we say {$\sum\limits_{i=0}^{n-1}A_{t,i}\mathbf{f}_i=0$} the $t$-th \textit{parity-check group}.

\subsection{The MDS property}
An $(n,k)$ array code defined by \eqref{Eqn parity check eq} is MDS if the source file can be reconstructed by connecting any $k$ out of the $n$ nodes. That is, any $r\times r$ sub-block matrix {$(A_{t,i})_{t\in [0: r),i\in J}$ of the block matrix $(A_{t,i})_{t\in [0: r),i\in[0: n)}$
is non-singular \cite{Barg1}, where $J$ is any $r$-subset of $[0: n)$.} In the following, we introduce some lemmas that will be helpful when verifying the MDS property of the new codes in the later sections.

\begin{Lemma}\label{lem the important for the proof of this section}
For $t, i\in [0:r)$, let $B_{t,i}$ be an $N\times N$ upper triangular matrix,  i.e.,
\begin{equation}\label{Eqn_upper}
B_{t,i}[a,b]=0 \mbox{~for~} 0\le b<a<N,
\end{equation}  
then
the block matrix $B=(B_{t,i})_{t\in [0: r),i\in[0: r)}$
   is non-singular
if  
\begin{itemize}
\item [i)] $B_{t,i}[a,a]=(B_{1,i}[a,a])^t$ for $i, t\in [0:r)$ and $a\in [0: N)$,

\item [ii)] $B_{1,i}[a,a]\ne B_{1,j}[a,a]$ for any $i,j\in [0:r)$ with $j\ne i$ and $a\in [0: N)$.
\end{itemize}

\end{Lemma}
\begin{proof}
The proof is given in Appendix \ref{sec:upper triangle lemma}.
\end{proof}

\subsection{Repair Mechanism}\label{sec:repair}

For an $(n,k)$ array code, suppose that node $i$ ($i\in[0: n)$) fails.
Let $H_i$ be any given $d$-subset of $[0: n)\setminus\{i\}$, which denotes  the  set of indices of the helper nodes, and let $L_i=[0: n)\setminus (H_i\cup\{i\})$
be the set of indices of unconnected nodes.
The data downloaded from helper node $j$ can be represented by $R_{i,j}\mathbf{f}_j$, where $R_{i,j}$ is a $\beta_{i,j} \times N$ matrix of full rank with $\beta_{i,j}\le N$. {We refer to $R_{i,j}$ as the \textit{repair matrix} of node $i$.}

Note that the content of node $i$ can be acquired from the parity-check equations.
In this paper, similar to \cite{eMSR_d_eq_n-1}, for convenience, we only consider the symmetric situation where $\delta$ ($N/r\le \delta\le N$) linearly independent equations are acquired from each of the  $r$ parity-check groups, where these $\delta$ linear independent equations are linear combinations of the corresponding $N$ parity-check equations in a parity-check group.
Precisely, the $\delta$  linear independent equations from the $t$-th parity-check group can be obtained by multiplying it with a $\delta \times N$ matrix $S_{i,t}$ of full rank, where $S_{i,t}$ is called the \textit{select  matrix}.
As a consequence, the following linear equations are available:

\begin{equation}\label{Eqn repairment gene}
\underbrace{\left(\begin{array}{c}
S_{i,0} A_{0,i} \\
S_{i,1} A_{1,i}\\
\vdots\\
S_{i,r-1} A_{r-1,i}
\end{array}\right)\mathbf{f}_i}_{\mathrm{useful ~data}}+\sum_{l\in L_i}\underbrace{\left(\begin{array}{c}
S_{i,0} A_{0,l}\\
S_{i,1} A_{1,l}\\
\vdots\\
S_{i,r-1} A_{r-1,l}
\end{array}\right)\mathbf{f}_l}_{\mathrm{unknown ~data~by~}\mathbf{f}_l}
+\sum_{j\in H_i}\underbrace{\left(\begin{array}{c}
S_{i,0} A_{0,j}\\
S_{i,1} A_{1,j}\\
\vdots\\
S_{i,r-1} A_{r-1,j}
\end{array}\right)\mathbf{f}_j}_{\mathrm{interference ~by~}\mathbf{f}_j}=0.
\end{equation}

The repair of node $i$ requires solving \eqref{Eqn repairment gene}  from the downloaded data  $R_{i,j}\mathbf{f}_{j}$, $j\in H_i$.
Then the repair bandwidth of node $i$ is
$
\gamma_i=\sum\limits_{j\in H_i}{\rm rank}(R_{i,j}).$
If $\gamma_i={d\over d-k+1}N$, then node $i$ is said to have the \textit{optimal repair bandwidth}, which can be accomplished if ${\rm rank}(R_{i,j})={N\over d-k+1}$ for all $j\in H_i$. If the repair bandwidth of an MDS array code is $(1+\epsilon){N\over d-k+1}$ where $\epsilon<1$ is a small constant, the MDS array code is said to have \textit{$(1+\epsilon)$-optimal repair bandwidth} in \cite{li2022pmds}.

\subsection{Partition of basis $\{e_0,\ldots,e_{N-1}\}$}

Assuming that $N=w^m$ for two integers $w$ and $m$ with $w,m\ge2$,  let $e_0,\ldots,e_{w^m-1}$ be a basis of $\mathbf{F}_q^{w^m}$.
For simplicity, one can regard them as the standard basis, i.e.,
\begin{equation*}
    e_i=(0,\ldots,0,1,0,\ldots,0),\,\,i\in [0: w^m),
\end{equation*}
with only the $i$-th entry being nonzero.

Then for any $a,b\in [0: N)$, we have
\begin{eqnarray}\label{Eqn e_ae_bT}
  e_a(e_b)^{\top} &=& \left\{
\begin{array}{ll}
1, & \textrm{if $a=b$},\\
0, & \textrm{otherwise},
\end{array}
\right.
\end{eqnarray}
where $\top$ represents the transpose operator.

In \cite{eMSR_d_eq_n-1}, a class of special partition sets of  $\{e_0,\ldots,e_{w^m-1}\}$ is given for $w\ge2$.  As these special partition sets play an important role in our proposed new construction, we revisit them for completeness in the following.

Given an integer $a\in [0: w^m)$, denote by $(a_0,\ldots,a_{m-1})$
its $w$-ary expansion with $a_0$ being the most significant digit, i.e., $a=\sum\limits_{j=0}^{m-1}w^{m-1-j}a_{j}$. For convenience, we also write $a=(a_0,\ldots,a_{m-1})$. 
For $i\in [0: m)$ and $t\in [0: w)$, define a subset of $\{e_0,\ldots,e_{w^m-1}\}$ as
\begin{equation}\label{Eqn_Vt}
V_{i,t}=\{e_a|a_i=t, 0\le a< w^m\},
\end{equation}
where $a_i$ is the $i$-th element in the $w$-ary expansion of $a$.

Obviously, $|V_{i,t}|=w^{m-1}$, and  $V_{i,0},V_{i,1},\ldots,V_{i,w-1}$ is a partition of the set $\{e_0,\ldots,e_{w^m-1}\}$ for any $i\in [0 : m)$.
Table \ref{example partition} gives two examples of the set partitions defined in \eqref{Eqn_Vt}.
\begin{table}[htbp]
\begin{center}
\caption{(a) and (b) denote the $m$  partition sets of $\{e_0,\ldots,e_{w^m-1}\}$    defined by \eqref{Eqn_Vt} for $m=3,w=2$, and $m=2,w=3$, respectively.}
\label{example partition}\begin{tabular}{|c|c|c|c|c|c|c|c|}
\hline $i$ & 0 & 1 & 2 & $i$ & 0 & 1 & 2\\
\hline \multirow{4}{*}{$V_{i,0}$ }&$e_0$&$e_0$&$e_0$&\multirow{4}{*}{$V_{i,1}$ }&$e_4$&$e_2$&$e_1$\\
  & $e_1$&$e_1$&$e_2$ && $e_5$&$e_3$&$e_3$\\
  &$e_2$&$e_4$&$e_4$&& $e_6$&$e_6$&$e_5$\\
  &$e_3$&$e_5$&$e_6$&& $e_7$&$e_7$&$e_7$\\
\hline\multicolumn{8}{c}{(A)}
\end{tabular}\hspace{5mm}
\begin{tabular}{|c|c|c|c|c|c|c|c|c|c}
\hline $i$ & 0 & 1 & $i$ & 0 & 1 & $i$ & 0 & 1\\
\hline \multirow{3}{*}{$V_{i,0}$ }&$e_0$&$e_0$&\multirow{3}{*}{$V_{i,1}$ }&$e_3$&$e_1$&\multirow{3}{*}{$V_{i,2}$ }&$e_6$&$e_2$\\
  & $e_1$&$e_3$&& $e_4$&$e_4$&& $e_7$&$e_5$\\
  & $e_2$&$e_6$&& $e_5$&$e_7$&& $e_8$&$e_8$\\
\hline\multicolumn{9}{c}{(B)}
\end{tabular}
\end{center}
\end{table}

For convenience of notation,  we also denote by $V_{i,t}$   the $w^{m-1}\times w^m$ matrix {whose rows are formed by vectors $e_a$
in their corresponding sets, and $a$  is sorted in ascending order}. For example, when $m=3$ and $w=2$, $V_{0,0}$ can be viewed as a $4\times 8$ matrix as follows
\begin{eqnarray*}
V_{0,0}=\left(e_0^{\top} e_1^{\top} e_2^{\top} e_3^{\top}\right)^{\top}.
\end{eqnarray*}

\subsection{Basic Notations and Equalities}
In this subsection, we introduce some useful notations and equalities that will facilitate the proof of the new codes. Let $N=w^m$, 
for $a=(a_0,\ldots,a_{m-1})\in [0: N)$, $i\in [0: m)$ and $u\in [0: w)$, define $a(i, u)$ as
\begin{equation}\label{Eqn_ait}
a(i,u)=(a_0,\ldots,a_{i-1},u,a_{i+1},\ldots,a_{m-1}),
\end{equation}
i.e., replacing the $i$-th digit by $u$.

For $a=(a_0,a_1,\ldots,a_{m-2})\in [0:N/w)$ and $i\in [0: m)$, define 
\begin{equation}\label{Eqn_g}
g_{i,u}(a)=(a_0,a_1,\ldots,a_{i-1},u,a_i,\ldots,a_{m-2}),
\end{equation}
i.e., inserting $u$ to the $i$-th digit of $(a_0,a_1,\ldots,a_{m-2})$.
Then for $i,j\in [0:m)$ and $u,v \in [0: w)$, we have that the $j$-th digit of $g_{i,u}(a)$ is 
\begin{eqnarray}\label{Eqn ga b}
(g_{i,u}(a))_j=\left\{
\begin{array}{ll}
a_j, &\mbox{~if~}j<i,\\
u,&\mbox{~if~}j=i,\\
a_{j-1},&\mbox{~if~}j>i.
\end{array}
\right.
\end{eqnarray}
Replacing the $j$-th digit of $g_{i,u}(a)$ by $v$ gives
\begin{eqnarray}\label{Eqn ga b (j,v)}
 (g_{i,u}(a))(j,v)=\left\{
\begin{array}{ll}
g_{i,u}(a(j,v)), &\mbox{~if~}j<i,\\
g_{i,v}(a),&\mbox{~if~}j=i,\\
g_{i,u}(a(j-1,v)),&\mbox{~if~}j>i.
\end{array}
\right.
\end{eqnarray}

  Let $e_0^{(N/w)}, e_1^{(N/w)}, \ldots, e_{N/w-1}^{(N/w)}$ be the standard basis vectors of $\mathbf{F}_q^{N/w}$ over $\mathbf{F}_q$, then by \eqref{Eqn_g}, $V_{i,u}$ in \eqref{Eqn_Vt} can be rewritten as
\begin{equation}\label{Eqn re Vu for liu}
  V_{i,u}=\sum\limits_{a=0}^{N/w-1}(e_a^{(N/w)})^{\top} e_{g_{i,u}(a)}, u\in [0:w),
\end{equation}
i.e., the $a$-th row of the matrix $V_{i,u}$ is
\begin{eqnarray}\label{Eqn re Vu[a] for liu}
   V_{i,u}[a,:]=e_{g_{i,u}(a)}, 0\le u<w,a\in [0:N/w).
\end{eqnarray}

{
\section{A new $(n,k)$ MSR code $\mathcal{C}_1$ with repair degree {$k<d<n$}}\label{sec:C3}
In this section,  we propose an $(n=2m,k=n-r)$ MSR code construction $\mathcal{C}_1$ with sub-packetization level $N=w^{m}$ and repair degree $d=k+w-1<n-1$ for some { $w\in [2: r+1)$. The new MSR code can be viewed as a combination of the YB code 1 in \cite{Barg1} and CB code in \cite{chen2019explicit}, i.e., half of the parity-check matrix of $\mathcal{C}_1$ is similar to the parity-check matrix of the YB code 1 while the other half is similar to that of CB code. This non-trivial combination leads to $\mathcal{C}_1$ having a larger code length or, equivalently, a smaller sub-packetization level than that of the CB code and YB code 1.}
Throughout this section, {let $c$ be a primitive element of  the finite field $\mathbf{F}_{q}$.}  

\begin{Construction}\label{Con_C0}
For $N=w^m$ and {$2\le w\le r$}, we define the parity-check matrix $(A_{t,i})_{t\in [0:r),i\in [0:n)}$ of the $(n=2m,k=n-r)$ array code $\mathcal{C}_1$ over $\mathbf{F}_{q}$ as
  \begin{eqnarray}\label{Eqn the PCM of code C0}
    A_{t,i}=\left\{\begin{array}{ll}
       \sum\limits_{a=0}^{N-1}\lambda_{i,a_i}^te_a^\top e_a+\sum\limits_{a=0,a_i=0}^{N-1}\sum\limits_{u=1}^{w-1}(\lambda_{i,0}^t-\lambda_{i,u}^t) e_a^\top e_{a(i,u)},  & \textrm{if~}i\in [0:m),  \\
        \sum\limits_{a=0}^{N-1}\lambda_{i,a_{i-m}}^t e_a^\top e_a,  & \textrm{if~}i\in [m:n),  \\ 
    \end{array}\right.
  \end{eqnarray}
where the repair degree is $d=k+w-1$, $\lambda_{i,j}\in \mathbf{F}_{q}$, {$a_i$ denotes the $i$-th digit of the $w$-ary expansion of $a$, and $\sum\limits_{a=0,a_i=0}^{N-1}$ denotes $a$ runs through all $[0: N)$ but with the restriction $a_i=0$.} We further define the repair matrix and select matrix of node $i$ as
  \begin{eqnarray}\label{Eqn the repair and select matrix of code C0}
     R_{i,j}=S_{i,t}=\left\{\begin{array}{ll}
          V_{i,0}, & \textrm{if~}i\in [0:m), \\
          V_{i,0}+V_{i,1}+\cdots+V_{i,w-1}, & \textrm{if~}i\in [m:n),
     \end{array}\right. 
  \end{eqnarray}
  for $t\in [0:r)$ and $j\in H_i$, where $H_i$ is any $d$-subset of $[0: n)\setminus\{i\}$, $V_{i,0},V_{i,1},\ldots,V_{i,w-1}$ for $i\in [0:m)$ are defined in \eqref{Eqn_Vt} and we further define
\begin{equation}\label{Eqn_Vt2}
V_{i,u}=V_{i-m,u}, \mbox{~for~}i\in  [m, n), u\in [0: w)
\end{equation}
for convenience of notation.
\end{Construction}

In what follows, we first give an example to show the connection between the new code and the YB code 1, CB code, and then anther example to show the main idea of this construction.

{
\begin{Example}
Consider the example where $r=3$, $w=2$, and $m=6$. In this case, let $(A_{t,i})_{t\in [0:3),i\in [0:12)}$ be the parity-check matrix of the $(12,9)$ code $\mathcal{C}_1$, then $(A_{t,i})_{t\in [0:3),i\in [0:6)}$ is exactly the parity-check matrix of the $(6,3)$ YB code 1 in \cite{Barg1} while $(A_{t,i})_{t\in [0:3),i\in [6:12)}$ is exactly the parity-check matrix of the $(6,3)$ CB code in \cite{chen2019explicit}.
\end{Example}
}

\begin{Example}\label{ex_C0}
An example of the $(6,3)$ MSR code $\mathcal{C}_1$ with sub-packetization level $8$ and repair degree $4$ over $\mathbf{F}_q$, where $q$ is any prime power larger than $12$. The parity-check matrix $(A_{t,i})_{t\in [0:3),i\in [0:6)}$ is defined as
\begin{equation*}
A_{t,0}=\begin{pmatrix}
        \lambda_{0,0}^te_0+ (\lambda_{0,0}^t-\lambda_{0,1}^t)e_4 \\
       \lambda_{0,0}^t e_1+(\lambda_{0,0}^t-\lambda_{0,1}^t)e_5 \\
       \lambda_{0,0}^t e_2+ (\lambda_{0,0}^t-\lambda_{0,1}^t)e_6 \\
       \lambda_{0,0}^t e_3  + (\lambda_{0,0}^t-\lambda_{0,1}^t)e_7 \\
        \lambda_{0,1}^te_4 \\
        \lambda_{0,1}^te_5 \\
        \lambda_{0,1}^te_6 \\
        \lambda_{0,1}^te_7 \\
     \end{pmatrix},~A_{t,1}=\begin{pmatrix}
        \lambda_{1,0}^te_0+ (\lambda_{1,0}^t-\lambda_{1,1}^t)e_2 \\
       \lambda_{1,0}^t e_1+(\lambda_{1,0}^t-\lambda_{1,1}^t)e_3 \\
       \lambda_{1,1}^t e_2 \\
       \lambda_{1,1}^t e_3   \\
        \lambda_{1,0}^te_4+ (\lambda_{1,0}^t-\lambda_{1,1}^t)e_6 \\
        \lambda_{1,0}^te_5+ (\lambda_{1,0}^t-\lambda_{1,1}^t)e_7 \\
        \lambda_{1,1}^te_6 \\
        \lambda_{1,1}^te_7 \\
     \end{pmatrix},
\end{equation*}    
\begin{equation*} 
 A_{t,2}=\begin{pmatrix}
        \lambda_{2,0}^te_0+ (\lambda_{2,0}^t-\lambda_{2,1}^t)e_1 \\
       \lambda_{2,1}^t e_1\\
       \lambda_{2,0}^t e_2+(\lambda_{2,0}^t-\lambda_{2,1}^t)e_3  \\
       \lambda_{2,1}^t e_3   \\
        \lambda_{2,0}^te_4+ (\lambda_{2,0}^t-\lambda_{2,1}^t)e_5 \\
        \lambda_{2,1}^te_5\\
        \lambda_{2,0}^te_6+ (\lambda_{2,0}^t-\lambda_{2,1}^t)e_7  \\
        \lambda_{2,1}^te_7 \\
    \end{pmatrix}, A_{t,3}=\begin{pmatrix}
        \lambda_{3,0}^te_0\\
       \lambda_{3,0}^t e_1\\
       \lambda_{3,0}^t e_2\\
       \lambda_{3,0}^t e_3\\
        \lambda_{3,1}^te_4 \\
        \lambda_{3,1}^te_5 \\
        \lambda_{3,1}^te_6 \\
        \lambda_{3,1}^te_7 \\
      \end{pmatrix},~A_{t,4}=\begin{pmatrix}
        \lambda_{4,0}^te_0\\
       \lambda_{4,0}^t e_1 \\
       \lambda_{4,1}^t e_2 \\
       \lambda_{4,1}^t e_3   \\
        \lambda_{4,0}^te_4 \\
        \lambda_{4,0}^te_5 \\
        \lambda_{4,1}^te_6 \\
        \lambda_{4,1}^te_7 \\
      \end{pmatrix},~A_{t,5}=\begin{pmatrix}
        \lambda_{5,0}^te_0 \\
       \lambda_{5,1}^t e_1\\
       \lambda_{5,0}^t e_2\\
       \lambda_{5,1}^t e_3   \\
        \lambda_{5,0}^te_4\\
        \lambda_{5,1}^te_5\\
        \lambda_{5,0}^te_6\\
        \lambda_{5,1}^te_7 \\
      \end{pmatrix},
\end{equation*}
where
\vspace{-5mm}
\begin{align}
\nonumber & \lambda_{0,0}=1, \lambda_{1,0}=c^4, \lambda_{2,0}=c^8, \lambda_{3,0}=c^2, \lambda_{4,0}=c^6, \lambda_{5,0}=c^{10},\\
\label{Eqn_ex_lamb_C0}& \lambda_{0,1}=c, \lambda_{1,1}=c^5, \lambda_{2,1}=c^9, \lambda_{3,1}=c^3, \lambda_{4,1}=c^7, \lambda_{5,1}=c^{11},
\end{align}
with $c$ being a primitive element in $\mathbf{F}_q$.

Suppose that Node $3$ fails and Node $0$ is not connected, we claim that Node $3$ can be repaired by connecting Nodes $1,2,4,5$ and downloading $(V_{0,0}+V_{0,1})\mathbf{f}_j$ (i.e., $(e_0+e_4)\mathbf{f}_{j}, (e_1+e_5)\mathbf{f}_{j}, (e_2+e_6)\mathbf{f}_{j}, (e_3+e_7)\mathbf{f}_{j}$) for $j=1,2,4,5$, and choose $S_{3,t}=V_{0,0}+V_{0,1}$ for $t=0,1,2$. Then, from \eqref{Eqn repairment gene}, we have
\begin{align}
\nonumber&   \begin{pmatrix}
           e_0+e_4 \\
       e_1+e_5\\
       e_2+e_6\\
       e_3+e_7  \\
        \lambda_{3,0}e_0+ \lambda_{3,1}e_4 \\
       \lambda_{3,0} e_1+\lambda_{3,1}e_5 \\
       \lambda_{3,0} e_2+ \lambda_{3,1}e_6 \\
       \lambda_{3,0} e_3  +\lambda_{3,1}e_7 \\
        \lambda_{3,0}^2e_0+ \lambda_{3,1}^2e_4 \\
       \lambda_{3,0}^2 e_1+\lambda_{3,1}^2e_5 \\
       \lambda_{3,0}^2 e_2+ \lambda_{3,1}^2e_6 \\
       \lambda_{3,0}^2 e_3  +\lambda_{3,1}^2e_7 \\
       \end{pmatrix}\mathbf{f}_3+ \begin{pmatrix}
           e_0+e_4 \\
       e_1+e_5\\
       e_2+e_6\\
       e_3+e_7  \\
             \lambda_{0,0}
             (e_0+e_4)\\
       \lambda_{0,0} (e_1+e_5)\\
              \lambda_{0,0} (e_2+e_6) \\
       \lambda_{0,0} (e_3+e_7)   \\
        \lambda_{0,0}^2(e_0+e_4) \\
       \lambda_{0,0}^2 (e_1+e_5)\\
              \lambda_{0,0}^2 (e_2+e_6) \\
       \lambda_{0,0}^2 (e_3+e_7)   \\
          \end{pmatrix}\mathbf{f}_0+\begin{pmatrix}
           e_0+e_4 \\
      e_1+e_5\\
       e_2+e_6\\
       e_3+e_7  \\
             \lambda_{1,0}(e_0+e_4)+(\lambda_{1,0}-\lambda_{1,1})(e_2+e_6) \\
     \lambda_{1,0}(e_1+e_5)+(\lambda_{1,0}-\lambda_{1,1})(e_3+e_7) \\
       \lambda_{1,1}(e_2+e_6)\\
       \lambda_{1,1}(e_3+e_7)\\
       \lambda_{1,0}^2(e_0+e_4)+(\lambda_{1,0}^2-\lambda_{1,1}^2)(e_2+e_6) \\
     \lambda_{1,0}^2(e_1+e_5)+(\lambda_{1,0}^2-\lambda_{1,1}^2)(e_3+e_7) \\
       \lambda_{1,1}^2(e_2+e_6)\\
       \lambda_{1,1}^2(e_3+e_7)
         \end{pmatrix}\mathbf{f}_1\\
\label{Eqn_ExC3}       =&-\begin{pmatrix}
           e_0+e_4 \\
      e_1+e_5\\
       e_2+e_6\\
       e_3+e_7  \\
             \lambda_{2,0}(e_0+e_4)+(\lambda_{2,0}-\lambda_{2,1})(e_1+e_5) \\
     \lambda_{2,1}(e_1+e_5)\\
       \lambda_{2,0}(e_2+e_6)+(\lambda_{2,0}-\lambda_{2,1})(e_3+e_7) \\
       \lambda_{2,1}(e_3+e_7)\\
       \lambda_{2,0}^2(e_0+e_4)+(\lambda_{2,0}^2-\lambda_{2,1}^2)(e_1+e_5) \\
     \lambda_{2,1}^2(e_1+e_5)\\
       \lambda_{2,0}^2(e_2+e_6)+(\lambda_{2,0}^2-\lambda_{2,1}^2)(e_3+e_7) \\
       \lambda_{2,1}^2(e_3+e_7)
           \end{pmatrix}\mathbf{f}_2
        -
         \begin{pmatrix}
           e_0+e_4 \\
      e_1+e_5\\
       e_2+e_6\\
       e_3+e_7  \\
        \lambda_{4,0}(e_0+e_4) \\
      \lambda_{4,0}(e_1+e_5)\\
       \lambda_{4,1}(e_2+e_6)\\
       \lambda_{4,1}(e_3+e_7)  \\
       \lambda_{4,0}^2(e_0+e_4) \\
      \lambda_{4,0}^2(e_1+e_5)\\
       \lambda_{4,1}^2(e_2+e_6)\\
       \lambda_{4,1}^2(e_3+e_7) 
          \end{pmatrix}\mathbf{f}_4-\begin{pmatrix}
           e_0+e_4 \\
      e_1+e_5\\
       e_2+e_6\\
       e_3+e_7  \\
        \lambda_{5,0}(e_0+e_4) \\
      \lambda_{5,1}(e_1+e_5)\\
       \lambda_{5,0}(e_2+e_6)\\
       \lambda_{5,1}(e_3+e_7)  \\
       \lambda_{5,0}^2(e_0+e_4) \\
      \lambda_{5,1}^2(e_1+e_5)\\
       \lambda_{5,0}^2(e_2+e_6)\\
       \lambda_{5,1}^2(e_3+e_7)  
      \end{pmatrix}\mathbf{f}_5,
\end{align}
which can be reformulated as
\begin{equation}\label{Eqn_ex_coef_M}
\underbrace{\begin{pmatrix}
I_4 & I_4 & I_4\\
\lambda_{3,0}I_4 & \lambda_{3,1}I_4 &\lambda_{0,0}I_4 \\ 
\lambda_{3,0}^2I_4 & \lambda_{3,1}^2I_4 &\lambda_{0,0}^2I_4 \\ 
\end{pmatrix}}_{M}\begin{pmatrix}
V_{0,0}\mathbf{f}_3\\
V_{0,1}\mathbf{f}_3\\
(V_{0,0}+V_{0,1})\mathbf{f}_0
\end{pmatrix}=\kappa_*,
\end{equation}
where $\kappa_*$ denotes
the data related to $\mathbf{f}_1,\mathbf{f}_2,\mathbf{f}_4,\mathbf{f}_5$ in \eqref{Eqn_ExC3} and can be determined from the downloaded data.

Using Lemma \ref{lem the important for the proof of this section} and \eqref{Eqn_ex_lamb_C0}, we can see that the matrix $M$ in \eqref{Eqn_ex_coef_M} is non-singular. Therefore, we can solve \eqref{Eqn_ex_coef_M} to obtain $V_{0,0}\mathbf{f}_3$ and $V_{0,1}\mathbf{f}_3$ (i.e., $\mathbf{f}_3$) and regenerate the lost data.

\end{Example}

In Example \ref{ex_C0}, it is obvious to see that all the matrices $A_{t,i}$, $t\in [0:3),i\in [0:6)$ are upper triangular. The situation also holds for the general case (cf. \eqref{Eqn the PCM of code C0}). Therefore, the MDS property can be easily verified according to Lemma \ref{lem the important for the proof of this section}. 
In the following, we formally analyze the MDS property of the new code $\mathcal{C}_1$.
\begin{Theorem}\label{Thr C0 is MDS}
The new code $\mathcal{C}_1$ is an MDS array code if
\begin{itemize}
 \item [i)] $\lambda_{i,u}\ne \lambda_{j,v}$ for $u, v\in [0: w)$ and $i,j\in [0: n)$ with $i\not\equiv j\bmod m$,
\item [ii)] $\lambda_{i,u}\ne \lambda_{i+m,u}$  for $u\in [0: w)$ and $i\in [0: m)$.
\end{itemize} 
\end{Theorem}
\begin{proof}
It suffices to prove that for any pairwise distinct $j_0,j_1,\ldots,j_{r-1}\in [0: n)$, the block matrix 
  \begin{equation}\label{Eqn_rrsubmC0}
   \begin{pmatrix}
          A_{0,j_0} &  A_{0,j_1} & \cdots & A_{0,j_{r-1}}\\
           A_{1,j_0} &  A_{1,j_1} & \cdots & A_{1,j_{r-1}}\\
           \vdots & \vdots & \vdots & \vdots\\
           A_{r-1,j_0} &  A_{r-1,j_1} & \cdots & A_{r-1,j_{r-1}}\\
      \end{pmatrix} 
  \end{equation}
  is non-singular over $\mathbf{F}_{q}$. 

 For any $a,b\in [0: N)$, $i\in [0:n)$ and $t\in [0: r)$, according to \eqref{Eqn the PCM of code C0}, we have
  \begin{eqnarray}\label{Eqn the value of A(a,b)}
    A_{t,i}[a,b]&=&e_aA_{t,i}e_b^\top\notag\\
    &=&\left\{\begin{array}{ll}
        e_a\left(\sum\limits_{z=0}^{N-1}\lambda_{i,z_i}^t e_z^\top e_z+\sum\limits_{z=0,z_i=0}^{N-1}\sum\limits_{u=1}^{w-1}(\lambda_{i,0}^t-\lambda_{i,u}^t)e_z^\top e_{z(i,u)}\right)e_b^\top,  & \textrm{if~}i\in [0: m),  \\
        e_a\left(\sum\limits_{z=0}^{N-1}\lambda_{i,z_{i-m}}^t e_z^\top e_z\right)e_b^\top, & \textrm{if~}i\in [m: n), \\
    \end{array}\right.\notag\\
        &=&\left\{\begin{array}{ll}
        \lambda_{i,a_i}^t e_a e_b^\top+\left(e_a\sum\limits_{z=0,z_i=0}^{N-1}\sum\limits_{u=1}^{w-1}(\lambda_{i,0}^t-\lambda_{i,u}^t)e_z^\top e_{z(i,u)}\right)e_b^\top,  & \textrm{if~}i\in [0: m), \\
        \lambda_{i,a_{i-m}}^te_ae_b^\top,  & \textrm{if~}i\in [m: n),
    \end{array}\right.\notag\\
 &=&\left\{\begin{array}{ll}
        \lambda_{i,a_i}^t,  & \textrm{if~}i\in [0: m), \textrm{ and } b=a,\\
        \lambda_{i,0}^t-\lambda_{i,u}^t, & \textrm{if~}i\in [0: m), a_i=0,\textrm{ and }  b=a(i,u) \textrm{~for~} u=1,2,\ldots,w-1, \\
        \lambda_{i,a_{i-m}}^t,  & \textrm{if~}i\in [m: n) \textrm{~and~}b=a,\\
        0,  & \textrm{otherwise,}
    \end{array}\right.
  \end{eqnarray}
which implies that $A_{t,i}[a,b]=0$ for $0\le b<a<N$ (i.e., $A_{t,i}$ is an upper triangular matrix) and
\begin{eqnarray}\label{Eqn the first useful equation for proving Thm 9}
  	A_{t,i}[a,a]=\lambda_{i,a_{i\%m}}^t=(A_{1,i}[a,a])^t \textrm{~for~} a\in [0: N),
  \end{eqnarray}
where $\%$ denotes the modulo operation, $t\in [0: r)$, and $i\in [0: n)$. {This implies that i) of Lemma \ref{lem the important for the proof of this section} holds for the matrix in \eqref{Eqn_rrsubmC0}.}

For any $t\in [0: r),a\in [0: N)$ and $0\le i< j<n$, by \eqref{Eqn the first useful equation for proving Thm 9}, 
we have 
\begin{equation*}
A_{1,i}[a,a]-A_{1,j}[a,a]=\left\{\begin{array}{ll}
       \lambda_{i,a_{i\%m}}-\lambda_{j,a_{j\%m}},  &\textrm{if~}i\not\equiv j \bmod m, \\
         \lambda_{i,a_{i\%m}}-\lambda_{j,a_{i\%m}},  & \textrm{otherwise,}
    \end{array}\right.  
\end{equation*}
 which together with i) and ii) implies $A_{1,i}[a,a]-A_{1,j}[a,a]\ne 0$, {i.e., ii) of Lemma \ref{lem the important for the proof of this section} holds for  the matrix in \eqref{Eqn_rrsubmC0}}.
Finally, applying Lemma \ref{lem the important for the proof of this section}, we claim that the matrix in \eqref{Eqn_rrsubmC0} is non-singular, and then we reach the desired result.
\end{proof}

Analyzing the repair property requires that \eqref{Eqn repairment gene} is solvable based on the downloaded data. Thus it is helpful to 
{characterize} the product of $S_{i,t}$ and $A_{t,j}$ beforehand. 
\begin{Lemma}\label{lem the form between S and A}
  For any $i,j\in [0:n)$,  rewrite them as
$i=g_0m+i'$ and $j=g_1m+j'$ for $g_0,g_1\in \{0,1\}$ and $i',j'\in [0: m)$. Then for $t\in [0:r)$, we have 
  \begin{itemize}
      \item [i)] $S_{i,t}A_{t,i}=\left\{\begin{array}{ll}
\lambda_{i,0}^t V_{i,0}+(\lambda_{i,0}^t-\lambda_{i,1}^t)V_{i,1}+\cdots+(\lambda_{i,0}^t-\lambda_{i,w-1}^t) V_{i,w-1}, &\textrm{if~} i\in [0: m),\\[4pt]
\lambda_{i,0}^t V_{i,0}+\lambda_{i,1}^t V_{i,1}+\cdots+\lambda_{i,w-1}^t V_{i,w-1}, &\textrm{if~} i\in [m:2m),      
 \end{array}\right.$
      \item [ii)] $S_{i,t}A_{t,j}=B_{t,j,i}R_{i,j}$ for $j\ne i$, where $B_{t,j,i}$ is an $\frac{N}{w}\times \frac{N}{w}$ matrix define by
      \begin{equation}\label{Eqn the bmatrix bar{B}}
        B_{t,j,i}=\left\{\begin{array}{ll}
           \sum\limits_{a=0}^{N/w-1}\lambda_{j,a_j}^t(e_a^{(N/w)})^\top e_a^{(N/w)}\\ \hspace{1cm}+\sum\limits_{a=0,a_j=0}^{N/w-1}\sum\limits_{u=1}^{w-1}(\lambda_{j,0}^t-\lambda_{j,u}^t)(e_a^{(N/w)})^\top e_{a(j,u)}^{(N/w)},  & \textrm{if~}j\in [0: i'),  \\
           \sum\limits_{a=0}^{N/w-1}\lambda_{j,a_{j-1}}^t(e_a^{(N/w)})^\top e_a^{(N/w)}\\ \hspace{1cm}+\sum\limits_{a=0,a_{j-1}=0}^{N/w-1}\sum\limits_{u=1}^{w-1}(\lambda_{j,0}^t-\lambda_{j,u}^t)(e_a^{(N/w)})^\top e_{a(j-1,u)}^{(N/w)},  & \textrm{if~}j\in [i'+1: m),  \\
       \sum\limits_{a=0}^{N/w-1}\lambda_{j,a_{j-m}}^t(e_a^{(N/w)})^\top e_a^{(N/w)},  & \textrm{if~}j\in [m: m+i'),  \\
           \sum\limits_{a=0}^{N/w-1}\lambda_{j,a_{j-m-1}}^t(e_a^{(N/w)})^\top e_a^{(N/w)},  & \textrm{if~}j\in [m+i'+1: n),  \\[6pt] 
           \lambda_{j,0}^t I_{N/w}, & \textrm{if~}j\equiv i\bmod m,
        \end{array}\right.
      \end{equation}
  \end{itemize}
  where $e_0^{(N/w)},e_1^{(N/w)},\ldots,e_{N/w-1}^{(N/w)}$ are the standard basis of $\mathbf{F}_{q}^{N/w}$.
 
 \item [iii)] For the matrix in \eqref{Eqn the bmatrix bar{B}}, we have
    \begin{eqnarray}\label{Eqn the first eq for proving Lem 10}
    	B_{t,j,i}[a,b]=0 \textrm{ for any } t\in [0,r),0\le b<a<N/w,
    \end{eqnarray}
    and
   \begin{eqnarray}\label{Eqn the second eq for proving Lem 10}
    	B_{t,j,i}[a,a]=\left\{\begin{array}{ll}
    		\lambda_{j,a_{j'}}^t, & \textrm{if }j'<i',\\
    		\lambda_{j,0}^t, & \textrm{if }j'= i',\\
    		\lambda_{j,a_{j'-1}}^t, & \textrm{if }j' >i',		
    	\end{array}\right.
    \end{eqnarray}
    for any $a\in [0: N/w)$.  
\end{Lemma}
\begin{proof}
The proof is given in Appendix \ref{sec:C0repair}.
\end{proof}

With this lemma, we can now analyze the repair property according to \eqref{Eqn repairment gene}.
\begin{Theorem}\label{Thm_C0 repair}
The new code $\mathcal{C}_1$ the optimal repair bandwidth with repair degree $d=k+w-1$ if
\begin{itemize}
\item [i)] $\lambda_{i,u}\ne \lambda_{i,v}$  for $u, v\in [0: w)$ with $u\ne v$ and $i\in [0: n)$, 
\item [ii)] $\lambda_{i,u}\ne \lambda_{j,v}$ for $u, v\in [0: w)$ and $i,j\in [0: n)$ with $i\not\equiv j\bmod m$,
\item [iii)] {If $w<r$}, $\lambda_{i,0}\ne \lambda_{i+m,u}$ and  $\lambda_{i,u}\ne \lambda_{i+m,0}$ for $u\in [0: w)$ and $i\in [0: m)$.
\end{itemize}  
\end{Theorem}
\begin{proof}
We consider the repair of node $i$ when $w<r$, where we only check $i\in [0:m)$ since the case $i\in [m:n)$ can be verified similarly. By Lemma \ref{lem the form between S and A}, we can express \eqref{Eqn repairment gene} as
\begin{equation}\label{Eqn_Thm_C0 repair}
   \left(\hspace{-1.5mm}\begin{array}{c}
            V_{i,0}\\  
	            \lambda_{i,0}V_{i,0}+\sum\limits_{t=1}^{w-1}(\lambda_{i,0}-\lambda_{i,t})V_{i,t}\\ 
            \vdots\\
            \lambda_{i,0}^{r-1}V_{i,0}+\sum\limits_{t=1}^{w-1}(\lambda_{i,0}^{r-1}-\lambda_{i,t}^{r-1})V_{i,t}\\ 
       \end{array}\hspace{-1.5mm}\right)\mathbf{f}_i+\sum\limits_{l\in L_i}
         \left(\hspace{-1.5mm}\begin{array}{c}
             B_{0,l,i}\\
             B_{1,l,i}\\
             \vdots\\
             B_{r-1,l,i}\\
          \end{array}\hspace{-1.5mm}\right)R_{i,l}\mathbf{f}_l+\sum\limits_{j\in H_i}
         \left(\hspace{-1.5mm}\begin{array}{c}
             B_{0,j,i}\\
             B_{1,j,i}\\
             \vdots\\
             B_{r-1,j,i}\\
          \end{array}\hspace{-1.5mm}\right)R_{i,j}\mathbf{f}_j=\mathbf{0},
\end{equation}
Let $L_i=\{l_0,l_1,\ldots,l_{r-w-1}\}$, substituting them into the above equations, we then have
\begin{eqnarray}\label{Eqn proving R3 of C0}
\nonumber  &&\hspace{-1.5mm}\underbrace{\left(\hspace{-1.8mm}\begin{array}{ccccccc}
    I_{N/w} &  0_{N/w} &\cdots & 0_{N/w} & B_{0,l_0,i} &\cdots & B_{0,l_{r-w-1},i}\\
    \lambda_{i,0}I_{N/w} & (\lambda_{i,0}-\lambda_{i,1})I_{N/w} & \cdots & (\lambda_{i,0}-\lambda_{i,w-1})I_{N/w} & B_{1,l_0,i} &\cdots& B_{1,l_{r-w-1},i}\\
    \vdots & \vdots & \ddots & \vdots & \vdots &\ddots & \vdots \\
    \lambda_{i,0}^{r-1}I_{N/w} &  (\lambda_{i,0}^{r-1}-\lambda_{i,1}^{r-1})I_{N/w} & \cdots & (\lambda_{i,0}^{r-1}-\lambda_{i,w-1}^{r-1})I_{N/w} & B_{r-1,l_0,i}&\cdots & B_{r-1,l_{r-w-1},i}\\
\end{array}\hspace{-1.8mm}\right)}_{\mathrm{B}}\hspace{-1.5mm}\\&&{\cdot} \left(\hspace{-1.8mm}\begin{array}{c}
   V_{i,0}\mathbf{f}_i\\
   \vdots\\
   V_{i,w-1}\mathbf{f}_i\\
   R_{i,l_0}\mathbf{f}_{l_0}\\
   \vdots\\
   R_{i,l_{r-w-1}}\mathbf{f}_{l_{r-w-1}}\\
\end{array}\hspace{-1.8mm}\right)=-\sum\limits_{j\in H_i}
         \left(\begin{array}{c}
             B_{0,j,i}\\
             B_{1,j,i}\\
             \vdots\\
             B_{r-1,j,i}\\
          \end{array}\right)R_{i,j}\mathbf{f}_j.
\end{eqnarray}
It is easy to see that the matrix $B$ can be converted to
\begin{equation}\label{Eqn_matB'}
B'=\left(\begin{array}{ccccccc}
    I_{N/w} & I_{N/w} &\cdots & I_{N/w} & B_{0,l_0,i} &\cdots & B_{0,l_{r-w-1},i}\\
    \lambda_{i,0}I_{N/w} & \lambda_{i,1}I_{N/w} & \cdots & \lambda_{i,w-1}I_{N/w} & B_{1,l_0,i} &\cdots& B_{1,l_{r-w-1},i}\\
    \vdots & \vdots & \ddots & \vdots & \vdots &\ddots & \vdots \\
    \lambda_{i,0}^{r-1}I_{N/w} &  \lambda_{i,1}^{r-1}I_{N/w} & \cdots & \lambda_{i,w-1}^{r-1}I_{N/w} & B_{r-1,l_0,i}&\cdots & B_{r-1,l_{r-w-1},i}\\
\end{array}\right)
\end{equation}
by elementary column {operations}.

By Lemma \ref{lem the form between S and A}-iii), we have that the matrices $B_{1,l_0,i}, \ldots, B_{1,l_{r-w-1},i}$ are upper triangular and
$B_{t,l_j,i}[a,a]=(B_{1,l_j,i}[a,a])^t{,}$ 
for $t\in [0: r)$, $j\in [0: r-w)$, and $a\in [0:N/w)$.
Similar to the proof of Theorem \ref{Thr C0 is MDS}, by  Lemma \ref{lem the important for the proof of this section} and \eqref{Eqn the second eq for proving Lem 10}, we easily have that the block matrix $B'$ in  \eqref{Eqn_matB'} is non-singular if
\begin{equation*}
\lambda_{i,0}, \lambda_{i,1}, \ldots, \lambda_{i,w-1}, B_{1,l_0,i}[a, a], \ldots, B_{1,l_{r-w-1},i}[a, a]{,}
\end{equation*}
are pairwise distinct for any $i\in [0: m)$, $l_0,\ldots,l_{r-w-1}\in L_i$, and $a\in [0: N)$, i.e.,
\begin{equation*}
\lambda_{i,0}, \lambda_{i,1}, \ldots, \lambda_{i,w-1}, B_{1,j,i}[a, a], j\in [0: n)\setminus\{i\}{,}
\end{equation*}
are pairwise distinct for any $i\in [0: m)$ and $a\in [0: N)$ since $L_i$ is an arbitrary $(r-w)$-subset of $[0: n)\setminus\{i\}$, which can be satisfied if i)-iii) hold according to \eqref{Eqn the second eq for proving Lem 10}. {Therefore, if conditions i)-iii) hold, then $B$ in \eqref{Eqn proving R3 of C0} is non-singular. As a result, we can solve for $V_{i,0}\mathbf{f}_i,\cdots,V_{i,w-1}\mathbf{f}_i$ (i.e., $\mathbf{f}_i$) and  $R_{i,l}\mathbf{f}_l,l\in L_i$, since the right side hand of \eqref{Eqn proving R3 of C0} is known from the downloaded data.}  

{When $w=r$, the proof is similar to the case $w<r$ instead that the condition in iii) is not needed by noting   $L_i=\emptyset$ in \eqref{Eqn_Thm_C0 repair}.}
\end{proof}

%

\begin{Theorem}\label{Thm_MSR_field}
The requirements in items i), ii) of Theorem \ref{Thr C0 is MDS} and i)--iii) of Theorem \ref{Thm_C0 repair} can be fulfilled by setting
 { \begin{eqnarray}\label{Eqn the coefficient lambda}
    \lambda_{i,u}=\left\{\begin{array}{ll}
c^{i(w+2)+u},& {\rm if~} w=2,\\
c^{i(w+1)+u},& {\rm if~} w\in [3:r), \\
c^{iw+u},& {\rm if~} w=r,
\end{array}\right.    \lambda_{i+m,u}=\left\{\begin{array}{ll}
c^{i(w+2)+w+u},& {\rm if~} w=2, \\
c^{i(w+1)+w},& {\rm if~} w\in [3:r), u=0, \\
c^{i(w+1)+u\% (w-1)+1},& {\rm if~} w\in [3:r), u\ge1, \\
c^{iw+(u+1)\% r},& {\rm if~} w=r,
\end{array}\right.
  \end{eqnarray}
 for $i\in [0:m)$ and $u\in [0:w)$,  }
{ where $c$ is a primitive element of 
$\mathbf{F}_q$} with
{\begin{equation*}
q>\left\{\begin{array}{ll}
m(w+2),& {\rm if~} w=2, \\
m(w+1),& {\rm if~} w\in [3:r), \\
mw,& {\rm if~} w=r. 
\end{array}\right.
\end{equation*}
}
\end{Theorem}
\begin{proof}
We only verify the case $w\in [3:r)$, as the proofs for the remaining cases are similar.

For any $i,j \in[0: n)$ and $u,v\in [0: w)$ with $(i,u)\ne (j, v)$, rewrite $i=g_0m+i'$ and $j=g_1m+j'$, where $g_0,g_1\in \{0,1\}$ and $i',j'\in [0: m)$.
\begin{itemize}
\item [i)] When $i\not\equiv j\% m$, by \eqref{Eqn the coefficient lambda}, we have
$\lambda_{i,u}=c^{i'(w+1)+t}$  and $\lambda_{j,v}=c^{j'(w+1)+s}$
for some $t,s\in [0: w+1)$. Then 
$
 \lambda_{i,u}-\lambda_{j,v}=c^{i'(w+1)+t}(1-c^{(j'-i')(w+1)+s-t})\ne 0
$
since 
\begin{equation*}
  0<|(j'-i')(w+1)+s-t|\le (m-1)(w+1)+w=m(w+1)-1<q-1.
\end{equation*}
Therefore, i) of Theorem \ref{Thr C0 is MDS} and also ii) of Theorem \ref{Thm_C0 repair} are satisfied. 

\item [ii)] For $i\in [0: m)$, by \eqref{Eqn the coefficient lambda}, we have
\begin{equation*}
\lambda_{i+m,0}=c^{i(w+1)+w}=c^wc^{i(w+1)}=c^w\lambda_{i,0}\ne \lambda_{i,0},
\end{equation*}
\begin{equation*}
\lambda_{i+m,u}=c^{i(w+1)+u\% (w-1)+1}\ne c^{i(w+1)+u}=\lambda_{i,u} \mbox{~for~} u\ge1
\end{equation*}
since $u\% (w-1)+1\ne u$ for $u\in [1: w)$,
which shows that ii) of Theorem \ref{Thr C0 is MDS} is satisfied. 

\item [iii)] From \eqref{Eqn the coefficient lambda}, it is obvious to see $\lambda_{i,u}\ne \lambda_{i,v}$ for $u\ne v$,
 i.e., i) of Theorem \ref{Thm_C0 repair} is satisfied.  

\item [iv)] For $u\in [1: w)$ and $i\in [0: m)$, by \eqref{Eqn the coefficient lambda}, we have 
\begin{equation*}
\lambda_{i,0}-\lambda_{i+m,u}=c^{i(w+1)}-c^{i(w+1)+(u\% (w-1))+1}\ne 0,
\end{equation*}
since $c^{(u\% (w-1))+1}\ne 1$
and $
\lambda_{i,u}-\lambda_{i+m,0}=c^{i(w+1)+u}-c^{i(w+1)+w}\ne 0.$
Note that $\lambda_{i,0}\ne \lambda_{i+m,0}$ has been proved in ii),
thus, iii) of Theorem \ref{Thm_C0 repair} is satisfied.  
\end{itemize}
This completes the proof.
\end{proof}

\begin{Remark}
In Construction \ref{Con_C0}, we assumed that $2\mid n$ for the $(n, k)$ MSR code $\mathcal{C}_1$. {If $2\nmid n$, through shortening, one can easily obtain an $(n,k)$ MSR code with repair degree $d$ from an $(n+1,k+1)$ MSR code $\mathcal{C}_1$ with repair degree $d+1$ \cite[Theorem 6]{Rashmi2011optimal}.}
\end{Remark}

\begin{Remark}
When $w=r$, since Theorem \ref{Thm_C0 repair}-iii) is not needed to satisfy, then we can choose $\lambda_{i+m,0}, \lambda_{i+m,1}, \ldots, \lambda_{i+m,w-1}$ from the set $\{\lambda_{i,0}, \lambda_{i,1}, \ldots, \lambda_{i,w-1}\}$, which leads to a smaller finite field compared to the case $w<r$.    
\end{Remark}

%

\section{A new MDS array code $\mathcal{C}_2$ with small sub-packetization level}\label{sec:eMSRd=n-1} 
In \cite{eMSR_d_eq_n-1}, a generic transformation that can transform any $(n', k')$ MSR code into a new $(n=sn',k)$ MDS array code was proposed { for any $s\ge2$}, which can greatly reduce the sub-packetization level by sacrificing a bit repair bandwidth. Note that the sub-packetization level of the base code determines that of the new array code. Thus, it is desirable to choose an MSR code with a small sub-packetization level as the base code. The MSR code $\mathcal{C}_1$ in the previous section has a small sub-packetization and is suitable to serve as the base code.
In this section, by applying the generic transformation in \cite{eMSR_d_eq_n-1} to the $(n', k')$ code $\mathcal{C}_1$ with $d'=n'-1$ in the previous section, we construct an $(n=sn',k)$ MDS array code $\mathcal{C}_2$ with small sub-packetization level and $(1+\epsilon)$-optimal repair bandwidth, where the repair degree is $d=n-1$. 

\begin{Construction}\label{Con C3} Based on the generic transformation in \cite{eMSR_d_eq_n-1}, the new $(n,k)$ array code $\mathcal{C}_2$ is constructed throught two steps as follows.
\begin{itemize}
\item [] \textbf{Step 1.} Choosing the $(n', k')$ MSR code $\mathcal{C}_1$ with repair degree $d'=n'-1$ in Section \ref{sec:C3} as the base code. 
Let $(A_{t,i'})_{t\in [0: r),i'\in[0: n')}$,  $S'_{i',t}$, and $R'_{i',j'}$ denote its parity-check matrix, select matrices, and repair matrices, where $r=n'-k'$, $i',j'\in [0: n')$, $j'\ne i'$, and $t\in [0: r)$.  

\item  [] \textbf{Step 2.}
Applying the generic transformation in \cite{eMSR_d_eq_n-1} to the $(n', k')$ MSR code $\mathcal{C}_1$, 
then an $(n=sn',k)$ array code $\mathcal{C}_2$ with repair degree $d=n-1$ is obtained,  where the parity-check matrix $(A_{t,i})_{t\in [0: r),i\in[0: n)}$, select matrices $S_{i,t}$, and repair matrices $R_{i,j}$ are given as
\begin{equation}\label{Eqn_C3_codingM}
A_{t,i}=x_{t,i}A'_{t,i\%n'},~S_{i,t}=S'_{i\%n',t},
R_{i,j}=\left\{
\begin{array}{ll}
R'_{i\%n',j\%n'}, & \textrm{if}~j\not\equiv i \bmod n',\\
I, & \textrm{otherwise},
\end{array}
\right.
\end{equation} 
with $i,j\in [0:n)$, $j\ne i$, $t\in [0, r)$, $x_{t,i}\in \mathbf{F}_q\setminus\{0\}$, and again $\%$ denotes the modulo operation. 
\end{itemize}
\end{Construction}

\begin{Lemma}(\cite[Theorem 2]{eMSR_d_eq_n-1})\label{Lemma_dupli_tran}
Every failed node of the new $(n, k)$ array code  $\mathcal{C}_2$
obtained by the generic transformation can be regenerated by
the repair matrices defined in \eqref{Eqn_C3_codingM}, the repair bandwidth is $(1+\frac{(s-1)(r-1)}{n-1})\gamma_{\rm optimal}$, where $\gamma_{\rm optimal}=\frac{n-1}{r}N$ denotes the optimal repair bandwidth.
\end{Lemma}

\begin{Example}\label{ex_C3}
From Construction \ref{Con_C0}, the parity-check matrix $(A_{t,i})_{t\in [0:2),i\in [0:4)}$ of the $(n'=4, k'=2)$ MSR code $\mathcal{C}_1$ with $d'=3$ and $N=(d'-k'+1)^{n'/2}=2^2$ is given as
\begin{equation*}
A_{t,0}=\begin{psmallmatrix}
        \lambda_{0,0}^te_0+ (\lambda_{0,0}^t-\lambda_{0,1}^t)e_2 \\
       \lambda_{0,0}^t e_1+(\lambda_{0,0}^t-\lambda_{0,1}^t)e_3 \\
        \lambda_{0,1}^te_2 \\
        \lambda_{0,1}^te_3 
      \end{psmallmatrix}\hspace{-1mm}, A_{t,1}=\begin{psmallmatrix}
        \lambda_{1,0}^te_0+ (\lambda_{1,0}^t-\lambda_{1,1}^t)e_1 \\
       \lambda_{1,1}^t e_1\\
       \lambda_{1,0}^t e_2+(\lambda_{1,0}^t-\lambda_{1,1}^t)e_3  \\
       \lambda_{1,1}^t e_3   
      \end{psmallmatrix}\hspace{-1mm},
A_{t,2}=\begin{psmallmatrix}
        \lambda_{2,0}^te_0\\
       \lambda_{2,0}^t e_1\\
        \lambda_{2,1}^te_2 \\
        \lambda_{2,1}^te_3 \\
      \end{psmallmatrix}\hspace{-1mm},
 A_{t,3}=\begin{psmallmatrix}
        \lambda_{3,0}^te_0 \\
       \lambda_{3,1}^t e_1\\
       \lambda_{3,0}^t e_2\\
       \lambda_{3,1}^t e_3  
      \end{psmallmatrix}, 
      \end{equation*}
where $t\in [0: 2)$, $\lambda_{t,i}$, $t\in [0: 2)$,  $i\in [0: 4)$ are set according to \eqref{Eqn the coefficient lambda}. By setting $s=2$ in Construction \ref{Con C3}, we obtain an $(n=8, k=6)$ MDS array code $\mathcal{C}_2$ with $N=2^2$ and $d=7$, based on the $(n'=4, k'=2)$ MSR code $\mathcal{C}_1$ from Construction \ref{Con_C0}. The parity-check matrix $(A_{t,i})_{t\in [0:2),i\in [0:8)}$ of the new MDS array code  $\mathcal{C}_2$ is given as
\begin{equation*}
A_{t,0}=\begin{pmatrix}
        \lambda_{0,0}^te_0+ (\lambda_{0,0}^t-\lambda_{0,1}^t)e_2 \\
       \lambda_{0,0}^t e_1+(\lambda_{0,0}^t-\lambda_{0,1}^t)e_3 \\
        \lambda_{0,1}^te_2 \\
        \lambda_{0,1}^te_3 
      \end{pmatrix},~A_{t,1}=\begin{pmatrix}
        \lambda_{1,0}^te_0+ (\lambda_{1,0}^t-\lambda_{1,1}^t)e_1 \\
       \lambda_{1,1}^t e_1\\
       \lambda_{1,0}^t e_2+(\lambda_{1,0}^t-\lambda_{1,1}^t)e_3  \\
       \lambda_{1,1}^t e_3   
      \end{pmatrix},~A_{t,2}=\begin{pmatrix}
        \lambda_{2,0}^te_0\\
       \lambda_{2,0}^t e_1\\
        \lambda_{2,1}^te_2 \\
        \lambda_{2,1}^te_3 \\
      \end{pmatrix},
\end{equation*}
\begin{equation*}
 A_{t,3}=\begin{pmatrix}
        \lambda_{3,0}^te_0 \\
       \lambda_{3,1}^t e_1\\
       \lambda_{3,0}^t e_2\\
       \lambda_{3,1}^t e_3  
      \end{pmatrix},~A_{t,4}=c^{4t}\begin{pmatrix}
        \lambda_{0,0}^te_0+ (\lambda_{0,0}^t-\lambda_{0,1}^t)e_2 \\
       \lambda_{0,0}^t e_1+(\lambda_{0,0}^t-\lambda_{0,1}^t)e_3 \\
        \lambda_{0,1}^te_2 \\
        \lambda_{0,1}^te_3 
      \end{pmatrix},~A_{t,5}=c^{4t}\begin{pmatrix}
        \lambda_{1,0}^te_0+ (\lambda_{1,0}^t-\lambda_{1,1}^t)e_1 \\
       \lambda_{1,1}^t e_1\\
       \lambda_{1,0}^t e_2+(\lambda_{1,0}^t-\lambda_{1,1}^t)e_3  \\
       \lambda_{1,1}^t e_3   
      \end{pmatrix},
\end{equation*}
\begin{equation*}
 A_{t,6}=c^{4t}\begin{pmatrix}
        \lambda_{2,0}^te_0\\
       \lambda_{2,0}^t e_1\\
        \lambda_{2,1}^te_2 \\
        \lambda_{2,1}^te_3 \\
      \end{pmatrix},~ A_{t,7}=c^{4t}\begin{pmatrix}
        \lambda_{3,0}^te_0 \\
       \lambda_{3,1}^t e_1\\
       \lambda_{3,0}^t e_2\\
       \lambda_{3,1}^t e_3  
      \end{pmatrix},
\end{equation*}
where
\begin{eqnarray}\label{Eqn_ex_lamb}
 \lambda_{0,0}=1, \lambda_{1,0}=c^2,   \lambda_{2,0}=c, \lambda_{3,0}=c^3, 
 \lambda_{0,1}=c, \lambda_{1,1}=c^3,  \lambda_{2,1}=1, \lambda_{3,1}=c^2,
\end{eqnarray}
with $c$ being a primitive element in $\mathbf{F}_q$ where $q>8$.
\end{Example}


\begin{Theorem}\label{Thm_C3}
Setting $x_{t,i}$ in \eqref{Eqn_C3_codingM} as
\begin{equation}\label{Eqn C3 x assi}
x_{t,i}=x_{i}^t{,}
\end{equation}
for some $x_i\in \mathbf{F}_q\setminus\{0\}$, where $i\in [0:n)$ and $t\in [0: r)$, the code $\mathcal{C}_2$ in Construction \ref{Con C3} is an $(n=sn', k)$ MDS array code with repair degree $d=n-1$ over $\mathbf{F}_q$ and repair bandwidth $(1+\frac{(s-1)(r-1)}{n-1})\gamma_{\rm optimal}$, where $\gamma_{\rm optimal}=\frac{n-1}{r}N$,
if the following conditions i)--iii) hold
\begin{itemize}
\item [i)] $x_i\lambda_{i',u}\ne x_j\lambda_{j',u'}$ for $u, u'\in [0: r)$ and $i,j\in [0: n)$ with $i\not\equiv j\bmod m$,
\item [ii)] $x_i\lambda_{i',u}\ne x_j\lambda_{j',u}$ for $u\in [0: r)$ and $i,j\in [0: n)$ with $i\ne j$ and $i\equiv j\bmod m$,
\item [iii)] $\lambda_{i',u}\ne \lambda_{i',u'}$ for $u,u' \in [0:r)$ with $u\ne u'$ and $i'\in [0: n)$,
\end{itemize}
where $i'=i\%n'$ and $j'=j\%n'$.
\end{Theorem}
 
\begin{proof}
The repair property follows from Lemma \ref{Lemma_dupli_tran}, and the proof of the MDS property is similar to that of Theorem \ref{Thr C0 is MDS}. Therefore, we omit it here. 
\end{proof}

\begin{Theorem}\label{Thm_C3_field}
The requirements in items i) - iii) of Theorem \ref{Thm_C3} can be fulfilled by setting $x_{i}=c^{\lfloor i/n' \rfloor mr}$ for $i\in [0:n)$, where $c$ is a primitive element of 
$\mathbf{F}_q$ with $q>smr$.
\end{Theorem}

\begin{proof}
For $i,j\in [0:n)$, we rewrite them as
$i=v_0n'+i'$ and $j=v_1n'+j'$ for $v_0,v_1\in [0: s)$ and $i',j'\in [0: n)$, and further rewrite $i'$ and $j'$ as $i'=g_0m+i''$ and $j'=g_1m+j''$, where $g_0,g_1\in \{0,1\}$ and $i'',j''\in [0: m)$. 
By \eqref{Eqn the coefficient lambda}, we have
  \begin{eqnarray}\label{Eqn the first eq for proving Thm. 10}
    x_i\lambda_{i',u}=c^{(v_0m+i'')r+(u+g_0)\% r}.
  \end{eqnarray}

Then, by \eqref{Eqn the first eq for proving Thm. 10}, items i) - iv) of Theorem \ref{Thm_C3} can be verified according to the following three cases.

\begin{itemize}
\item For $u, u'\in [0: w)$ and $i,j\in [0: n)$ with $i\not\equiv j\bmod m$, i.e., $i''\ne j''$, we have
\begin{equation*}
x_i\lambda_{i',u}- x_j\lambda_{j',u'}=c^{(v_0m+i'')r+(u+g_0)\% r}(1-c^{\big((v_1-v_0)m+j''-i''\big)r+(u'+g_1)\% r-(u+g_0)\% r})\ne 0
\end{equation*}
since 
$
0<|\big((v_1-v_0)m+j''-i''\big)r+(u'+g_1)\% r-(u+g_0)\% r|\le smr-1<q-1.
$
Then i) of Theorem \ref{Thm_C3} is satisfied.

\item  For $u\in [0: w)$ and $i,j\in [0: n)$ with $i\ne j$ and $i\equiv j\bmod m$, i.e, $i''=j''$, we have
\begin{equation*}
x_i\lambda_{i',u}- x_j\lambda_{j',u}=c^{(v_0m+i'')r+(u+g_0)\% r}(1-c^{(v_1-v_0)mr+(u+g_1)\% r-(u+g_0)\% r})\ne 0
\end{equation*}
since 
$
0<|(v_1-v_0)mr+(u+g_1)\% r-(u+g_0)\% r|\le (s-1)mr+1<q-1,
$
i.e., ii) of Theorem \ref{Thm_C3} is satisfied.

\item It is obvious that iii) of Theorem \ref{Thm_C3} is satisfied according to \eqref{Eqn the coefficient lambda}.
\end{itemize}
This completes the proof.
\end{proof}

\section{Comparisons}\label{sec:comp}
In this section, we provide a detailed comparison of some key parameters among the proposed $(n, k)$ MSR code $\mathcal{C}_1$ with repair degree $d<n-1$, $(n, k)$ MDS array code $\mathcal{C}_2$ with repair degree $d=n-1$, and existing ones.
Table \ref{Table comp MSR} provides the details of the comparison between the proposed $(n, k)$ MSR code $\mathcal{C}_1$ with repair degree $d<n-1$ and existing ones. {Meanwhile, Figure \ref{pic_comp} shows the sub-packetization levels and required field sizes of each code with a repair degree of $d=k+2$ when the code length ranges from $10$ to $100$.}

\begin{table}[htbp]
\begin{center}
\caption{A comparison of the key parameters of $(n,k)$ MSR codes with sub-packetization level $N$.}\label{Table comp MSR}
\setlength{\tabcolsep}{2.6pt}
\begin{tabular}{|c|c|c|c|c|}
\hline  
& $N$& Field size  & {Repair degree}  &  References \\
\hline
 YB code 1    &$w^{n}$   &  $q\ge wn$  &{ $d=k+w-1\in [k+1:n)$}  &  \cite[Section IV]{Barg1} \\ 
\hline
YB code 2      &$w^{n}$   &  $q>n$  &  {   $d=k+w-1\in [k+1:n)$ }&   \cite[Section VIII]{Barg1}\\
\hline
\multirow{2}{*}{VBK code}    &$w^{\lceil \frac{n}{w} \rceil}$   &\multirow{2}{*}{  $q\ge \hspace{-1mm}\left\{\hspace{-2mm}\begin{array}{ll}
6\lceil \frac{n}{2} \rceil+2, w=2\\
18\lceil \frac{n}{w} \rceil+2, w=3,4
\end{array} \right.$}  &  { \multirow{2}{*}{$d=k+1, k+2, k+3$}} &   \multirow{2}{*}{\cite{vajha2021small}}\\ 
&($w=2,3,4$)&&&\\[4pt]
\hline
CB code     &$w^{n}$   &  $q>n+w$ &{  $d=k+w-1\in [k+1:n)$}   &     \cite{chen2019explicit}\\
\hline
LLT code    &$w^{\lceil \frac{n}{2} \rceil}$  &$q>n+\lceil\frac{n}{2} \rceil w$  &{  $d=k+w-1\in [k+1:n-1)$}   &    \cite{liu2022generic}\\
\hline
ZZ code  &$w^{\lceil \frac{n}{2} \rceil}$ &    $q\ge wn$   & { $d=k+w-1\in [k+1:n)$}  & \cite{zhang2023vertical}  \\
\hline
New code $\mathcal{C}_1$    &$w^{\lceil \frac{n}{2} \rceil}$    &  $q>\left\{\hspace{-2mm}\begin{array}{ll}
4\lceil \frac{n}{2}\rceil, w=2\\
\lceil \frac{n}{2}\rceil (w+1), w\in [3: r)\\
\lceil \frac{n}{2}\rceil w, w=r
\end{array} \right.$  &  {$d=k+w-1\in [k+1:n)$}  & Theorem \ref{Thm_MSR_field}\\
\hline
\end{tabular}
\end{center}
\end{table}

\begin{figure*}[htbp]
\centering
\hspace{-5mm}
\begin{minipage}[t]{0.4\textwidth}
\includegraphics[scale=.7]{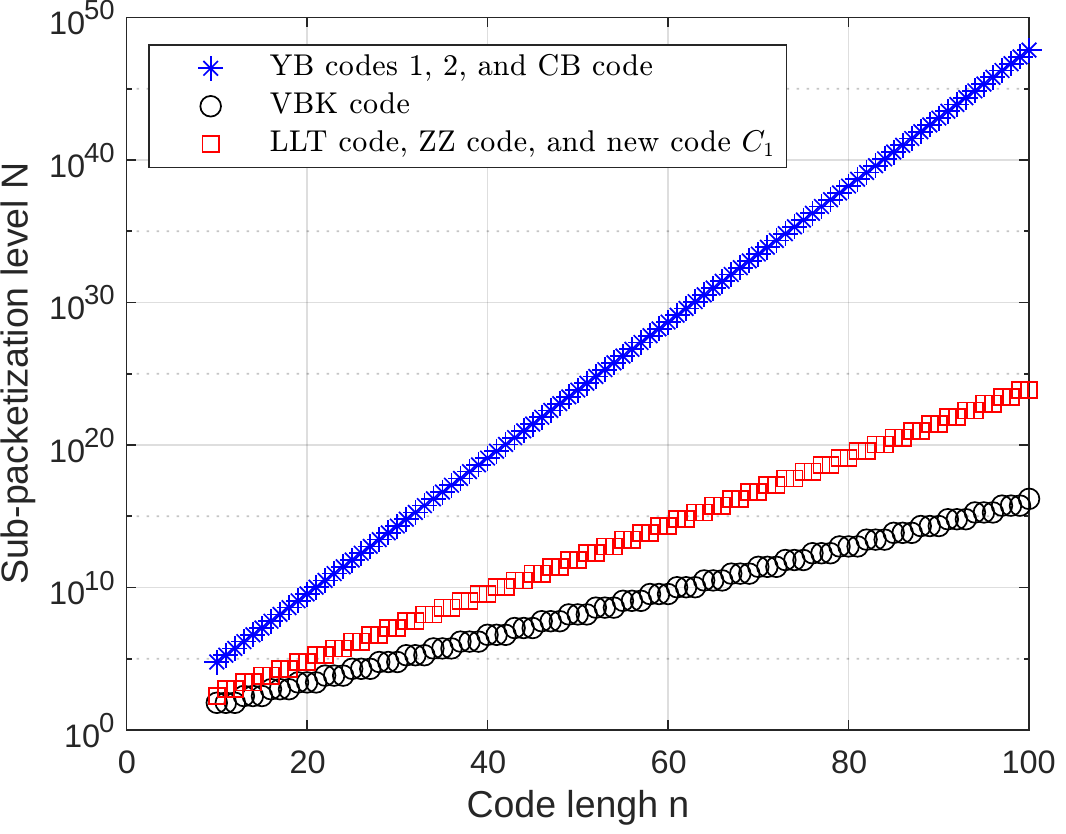}
\end{minipage}
\hspace{18mm}
\begin{minipage}[t]{0.4\textwidth}
\includegraphics[scale=.7]{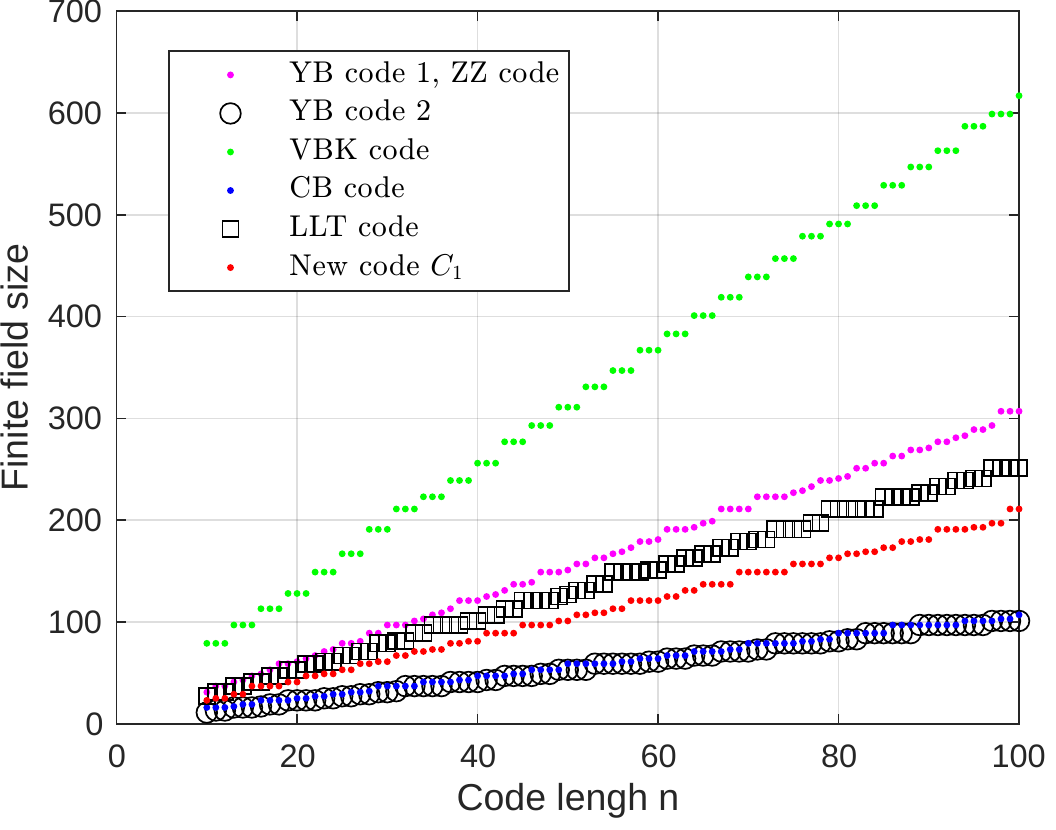}
\end{minipage}
\caption{Comparision of the sub-packetization level and finite field size among the new $(n, k)$ MSR code $\mathcal{C}_1$ and some known ones with repair degree $d=k+2$.}\label{pic_comp}
\end{figure*}

From Table \ref{Table comp MSR} and Figure \ref{pic_comp}, we see that the new MSR code $\mathcal{C}_1$ has the following advantages.
\begin{itemize}
\item [i)] The new MSR code $\mathcal{C}_1$ has a significantly smaller sub-packetization level than the YB codes 1, 2 in  \cite{Barg1}, and the CB code in \cite{chen2019explicit}, and a smaller finite field than that of YB code 1.

\item [ii)] The new MSR code $\mathcal{C}_1$ works for any repair degree $d\in [k+1 :n)$, which is much more flexible than that of the VBK code in \cite{vajha2021small}, which is restricted to $d=k+1, k+2, k+3$. Additionally, $\mathcal{C}_1$ requires a much smaller finite field than the VBK code when $w\in \{2,3,4\}$. {Specifically, when $w=2, 3$, and $4$, $\mathcal{C}_1$ requires a finite field $\mathbf{F}_q$ with size $q>4\lceil \frac{n}{2}\rceil$, $q>3\lceil \frac{n}{2}\rceil$, and $q>4\lceil \frac{n}{2}\rceil$, respectively. In contrast, the VBK code requires  a finite field $\mathbf{F}_q$ with size $q>6\lceil \frac{n}{2}\rceil+1$, $18\lceil \frac{n}{3}\rceil+1$, and $18\lceil \frac{n}{4}\rceil+1$, respectively.  However, it should be noted that when $d\in \{k+2,k+3\}$, the VBK code has a smaller sub-packetization level than that of the new code $\mathcal{C}_1$.}


\item [iii)] $\mathcal{C}_1$ has the same sub-packetization level as that of the LLT code in \cite{liu2022generic} and ZZ code in \cite{zhang2023vertical}. However, the $(n, k)$ LLT code does not work for $d=n-1$ and requires a larger finite field than $\mathcal{C}_1$, while the ZZ code requires a larger finite field than $\mathcal{C}_1$ when $d>k+1$. 

\item [iv)] The new MSR code $\mathcal{C}_1$ subsumes the YB code 1 in  \cite{Barg1} and the CB code in \cite{chen2019explicit} as subcodes, i.e.,  YB code 1 and  CB code can be obtained by shortening the new code $\mathcal{C}_1$.
\end{itemize}

Table \ref{Table comp_duplication} provides the details of the comparison between the proposed MDS array code $\mathcal{C}_2$ and existing ones with $(1+\epsilon)$-optimal repair bandwidth and repair degree $d=n-1$. {Figure \ref{pic_comp2} provides an additional example of the comparison of sub-packetization levels among the codes listed in Table \ref{Table comp_duplication}, with the exception of RTGE code 2 in \cite{Rawat}. This code relies on the existence of an error-correcting code with specific parameters, which may not always be available.}

\begin{table}[htbp]
\begin{center}
\caption{A comparison of the key parameters among the new $(n=sn',k)$ MDS array code $\mathcal{C}_2$ and existing ones with $(1+\epsilon)$-optimal repair bandwidth and repair degree $d=n-1$, where $\epsilon=\frac{(s-1)(r-1)}{n-1}$ and $r=n-k$.}\label{Table comp_duplication}
\begin{tabular}{|c|c|c|c|}
\hline
& {Sub-packatization $N$}& {Field size } &Repair bandwidth    \\
\hline
RTGE code 1 in \cite{Rawat} &   $r^{\lceil \frac{n'}{r}\rceil}$  & $q>n^{(r-1)N+1}$    & $(1+\epsilon)\gamma_{\rm optimal}$  \\
\hline
RTGE code 2 in \cite{Rawat} &   $O(r^{r\tau}\log n)$  & $O(n)$    & $\le (1+\frac{1}{\tau})\gamma_{\rm optimal}$    \\
\hline
MDS code $\mathcal{C}_1$ in \cite{eMSR_d_eq_n-1}  & $r^{n'}$ &  $q>rn'\lceil \frac{s}{r}\rceil$, $r\mid (q-1)$ {(i.e., $O(n)$)} & $(1+\epsilon)\gamma_{\rm optimal}$   \\
\hline
MDS code $\mathcal{C}_2$ in \cite{eMSR_d_eq_n-1}  & $r^{n'-1}$ &  $q>r\lceil \frac{n'}{r}\rceil(s-1)+n'$ {(i.e., $O(n)$)} & $(1+\epsilon)\gamma_{\rm optimal}$   \\
\hline
MDS code $\mathcal{C}_4$ in \cite{eMSR_d_eq_n-1}  & $r^{\lceil\frac{n'}{r+1}\rceil}$ &  $\begin{array}{ll}
q>\lceil\frac{2n}{3}\rceil, & {\rm if}~r=2\\
q>N {n-1\choose r-1}+1, & {\rm if}~r>2
\end{array}
$ & $(1+\epsilon)\gamma_{\rm optimal}$   \\
\hline
MDS code $\mathcal{C}_5$ in \cite{eMSR_d_eq_n-1}  &  $r^{n'}$ &  $q>rn'\lceil \frac{s}{r}\rceil$ {i.e., ($O(n)$)}& $(1+\epsilon)\gamma_{\rm optimal}$   \\
\hline
New MDS code $\mathcal{C}_2$  & $r^{\lceil \frac{n'}{2}\rceil}$ &  $q>sr\lceil \frac{n'}{2}\rceil$ {(i.e., $O(sn/2)$)}& $(1+\epsilon)\gamma_{\rm optimal}$  \\
\hline
\end{tabular}
\end{center}
\end{table}

\begin{figure*}[htbp]
\centering
\includegraphics[scale=.65]{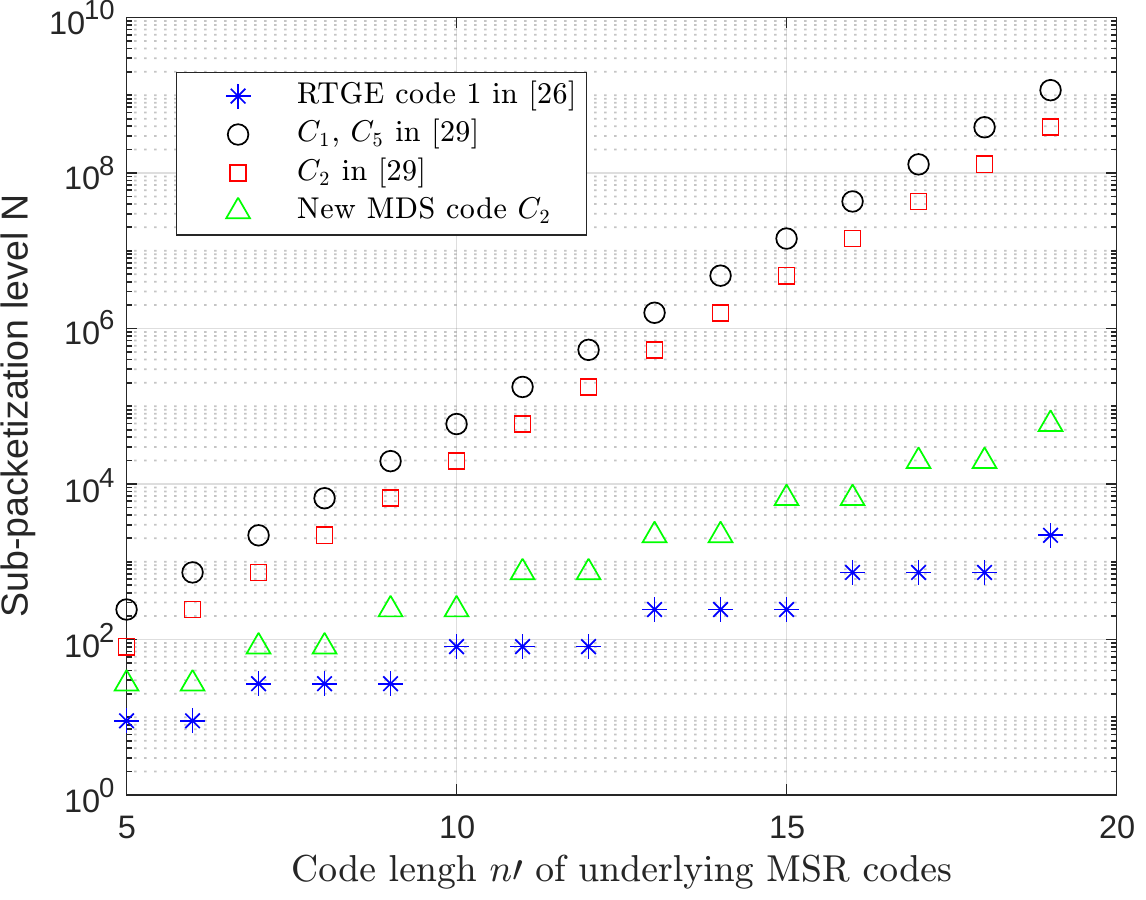}
\caption{Comparision of the sub-packetization level among the new $(n, k)$ MSR code $\mathcal{C}_1$ and some known ones with $r=3$.}\label{pic_comp2}
\end{figure*}

From Table \ref{Table comp_duplication} and {Figure \ref{pic_comp2}}, we can see that the proposed MDS array code $\mathcal{C}_2$ has the
following advantages compared to existing ones:

\begin{itemize}
\item Under the same repair bandwidth, the new MDS code $\mathcal{C}_2$ has a much smaller sub-packetization level when compared to the MDS array codes $\mathcal{C}_1$, $\mathcal{C}_2$, and $\mathcal{C}_5$ in \cite{eMSR_d_eq_n-1}. 
\item By noting that the MDS array code $\mathcal{C}_4$ in \cite{eMSR_d_eq_n-1} is implicit when $r>2$, we have that among all the explicit MDS array codes with $(1+\epsilon)$-optimal repair bandwidth and $r>2$, the new MDS code $\mathcal{C}_2$ has the smallest sub-packetization level under the same code parameters except for the RTGE code 1 in \cite{Rawat}, which requires a super large finite field.
\end{itemize}

\section{Conclusion}\label{sec:conclusion}
In this paper, we proposed a new $(n, k)$ MSR code construction that works for any repair degree $d>k$. The new MSR code has a smaller sub-packetization level or finite field than existing ones. Additionally, we obtained a new $(n,k)$ MDS array code with a small sub-packetization level, $(1+\epsilon)$-optimal repair bandwidth, and repair degree $d=n-1$, which outperforms existing ones in terms of the sub-packetization level or the field size. For $(n,k)$ MDS array code with small sub-packetization level, $(1+\epsilon)$-optimal repair bandwidth, and repair degree $d<n-1$, few results have been reported in the literature. To the best of our knowledge, the only one is the construction in \cite{Guruswami2020}, which only works for very large parameters $n, k$ and requires a huge finite field, thus {it is infeasible} to be implemented in practical systems. Constructions of $(n,k)$ MDS array code over small finite fields with small sub-packetization level, $(1+\epsilon)$-optimal repair bandwidth, and repair degree $d<n-1$ will be left for our future research.

\appendices
\section{Proof of Lemma \ref{lem the important for the proof of this section}}\label{sec:upper triangle lemma} 
For $i\in [0: rN)$, let $e_i$ be row $i$ of the identity matrix of order $rN$. Then define an $rN\times rN$ permutation matrix $\Psi$ as
\begin{equation*}
\Psi=(e_0^{\top}, e_N^{\top}, \ldots, e_{(r-1)N}^{\top}, e_1^{\top}, e_{1+N}^{\top}, \ldots, e_{1+(r-1)N}^{\top}, \ldots, e_{N-1}^{\top}, e_{N-1+N}^{\top}, \ldots, e_{N-1+(r-1)N}^{\top})^{\top},
\end{equation*}
where $\top$ denotes the transpose operator.

   Multiplying matrices $\Psi$ and $\Psi^\top$ on the left and right sides of matrix $B=(B_{t,i})_{t\in [0: r),i\in[0: r)}$, respectively, we then have
\begin{equation*}
 \Psi B\Psi^\top =\begin{pmatrix}
 B'_{0,0} &  B'_{0,1} & \cdots & B'_{0,N-1}\\
           B'_{1,0} &  B'_{1,1} & \cdots & B'_{1,N-1}\\
           \vdots & \vdots & \vdots & \vdots\\
           B'_{N-1,0} &  B'_{N-1,1} & \cdots & B'_{N-1,N-1}     
 \end{pmatrix},
\end{equation*}
where
\begin{equation*}
B'_{a,b}=\begin{pmatrix}
 B_{0,0}[a,b] & B_{0,1}[a,b] & \cdots & B_{0,r-1}[a,b]\\ 
B_{1,0}[a,b] & B_{1,1}[a,b] & \cdots & B_{1,r-1}[a,b]\\
 \vdots & \vdots & \ddots & \vdots \\
B_{r-1,0}[a,b] & B_{r-1,1}[a,b] & \cdots & B_{r-1,r-1}[a,b]
\end{pmatrix}, a,b\in [0:N).
\end{equation*}
By \eqref{Eqn_upper}, we have $B'_{a,b}=\bf{0}$ for $0\le b<a<N$.
For $a\in [0:N)$, by i), we have
\begin{equation}
B'_{a,a}= \begin{pmatrix}1 & 1 & \cdots &1\\ 
B_{1,0}[a,a] & B_{1,1}[a,a] & \cdots & B_{1,r-1}[a,a]\\
 \vdots & \vdots & \ddots & \vdots \\
B_{1,0}^{r-1}[a,a] & B_{1,1}^{r-1}[a,a] & \cdots & B_{1,r-1}^{r-1}[a,a]
\end{pmatrix},
\end{equation}
which is a Vandermonde matrix and is non-singular according to ii). 
Therefore, 
\begin{equation*}
  |\Psi||B||\Psi^\top |=|\Psi B\Psi^\top |=\left|\begin{array}{cccc}
          B'_{0,0} &  B'_{0,1} & \cdots & B'_{0,N-1}\\
           &  B'_{1,1} & \cdots & B'_{1,N-1}\\
            &   & \ddots & \vdots\\
             &  &  & B'_{N-1,N-1}\\
      \end{array}\right|=\prod\limits_{a=0}^{N-1}|B'_{a,a}|\ne 0,
\end{equation*}
which implies that $B$ is non-singular.

\section{Proof of Lemma \ref{lem the form between S and A}}\label{sec:C0repair}

\begin{proof}
Hereafter we only check the case of $i\in [0: m)$ since the other case can be proved similarly. For any given $a=(a_{0},a_{1},\ldots,a_{m-2})\in [0: N/w)$, $u\in [0: w)$, and $j\in [0: n)$, according to \eqref{Eqn the PCM of code C0}, we have\\
	\begin{align}\label{eqn e_aQti}
		V_{i,u}[a,:]A_{t,j}\hspace{-1mm}&=\hspace{-1mm}\left\{\hspace{-2mm}\begin{array}{ll}
        e_{g_{i,u}(a)}\left(\sum\limits_{b=0}^{N-1}\lambda_{j,b_j}^t e_b^\top e_b+\sum\limits_{b=0,b_j=0}^{N-1}\sum\limits_{v=1}^{w-1}(\lambda_{j,0}^t-\lambda_{j,v}^t)e_b^\top e_{b(j,v)}\right),  & \textrm{if~}j\in [0: m),  \\
        e_{g_{i,u}(a)}\left(\sum\limits_{b=0}^{N-1}\lambda_{j,b_{j-m}}^t e_b^\top e_b\right), & \textrm{if~}j\in [m:n), \\
    \end{array}\right.\notag\\
        \hspace{-3mm}&=\hspace{-1mm}\left\{\hspace{-2mm}\begin{array}{ll}
		   \lambda_{j,0}^t e_{g_{i,u}(a)}+\sum\limits_{v=1}^{w-1}   (\lambda_{j,0}^t-\lambda_{j,v}^t) e_{(g_{i,u}(a))(j,v)}, & \textrm{if~}j\in [0:m)\textrm{~and~}(g_{i,u}(a))_j=0,\\
              \lambda_{j,(g_{i,u}(a))_j}^t e_{g_{i,u}(a)}, & \textrm{if~}j\in [0: m)\textrm{~and~}(g_{i,u}(a))_j\ne 0,\\
		  \lambda_{j,(g_{i,u}(a))_{j-m}}^t e_{g_{i,u}(a)}, & \textrm{if~}j\in [m: n).\\
		\end{array}\right.
	\end{align}
where the two equalities follow from \eqref{Eqn re Vu[a] for liu} and \eqref{Eqn e_ae_bT}, respectively.  

By \eqref{Eqn ga b} and \eqref{Eqn ga b (j,v)}, we have $(g_{i,u}(a))_i=u$ and $e_{g_{i,u}(a)(i,v)}=e_{g_{i,v}(a)}$. Then by \eqref{Eqn re Vu[a] for liu} and \eqref{eqn e_aQti}, when $j\equiv i\bmod m$, i.e., $j=i$ or $j=i+m$, we have
	\begin{equation*}\label{eqn VtQti for VBK code eq 1}
		V_{i,u}[a,:]A_{t,j}=\left\{\begin{array}{ll}
			\lambda_{i,0}^tV_{i,0}[a,:]+\sum\limits_{v=1}^{w-1}(\lambda_{i,0}^t-\lambda_{i,v}^t)V_{i,v}[a,:], & \textrm{if~} j=i,u=0,\\
			\lambda_{j,u}^tV_{i,u}[a,:], & \textrm{if~}j=i,u\ne0 \textrm{~or~}j=i+m.
		\end{array}\right.
	\end{equation*}
	That is, for $j\equiv i\bmod m$ and $u\in [0: w)$, we have
	\begin{equation*}\label{eqn VtQti for VBK code}
		V_{i,u}A_{t,j}=\left\{\begin{array}{ll}
			\lambda_{i,0}^tV_{i,0}++\sum\limits_{v=1}^{w-1}(\lambda_{i,0}^t-\lambda_{i,v}^t)V_{i,v}, & \textrm{if~} j=i,u=0,\\
			\lambda_{j,u}^tV_{i,u}, & \textrm{if~}j=i,u\ne0 \textrm{~or~}j=i+m.
			\end{array}\right.
	\end{equation*}
	which together with \eqref{Eqn the repair and select matrix of code C0} implies  
	\begin{equation*}
		S_{i,t}A_{t,i}=V_{i,0}A_{t,i}=\lambda_{i,0}^tV_{i,0}+(\lambda_{i,0}^t-\lambda_{i,1}^t)V_{i,1}+\cdots+(\lambda_{i,0}^t-\lambda_{i,w-1}^t)V_{i,w-1}
	\end{equation*}
	and 
	$
		S_{i,t}A_{t,j}=\lambda_{j,0}^t R_{i,j}
	$  for  $j=i+m$,  
	i.e., i) is true and ii) holds for $i, j \in [0: n)$ with $j\ne i$ and $j\equiv i\bmod m$.
	
Next, we prove that ii) holds for $j\not\equiv i \bmod m$. Recall that we only check the case of $i\in [0: m)$, which is discussed in the following four cases.

\textit{Case 1.} If $j\in [0: i)$, by applying \eqref{Eqn ga b} and \eqref{Eqn ga b (j,v)} to \eqref{eqn e_aQti},  we then have
	\begin{eqnarray}\label{Eqn the product between V[a] and Atj}
	&&V_{i,u}[a,:]\cdot A_{t,j}\notag\\
		&=&\left\{\begin{array}{ll}
		   \lambda_{j,0}^t e_{g_{i,u}(a)}+\sum\limits_{v=1}^{w-1}(\lambda_{j,0}^t-\lambda_{j,v}^t) e_{g_{i,u}(a(j,v))}, & \textrm{if~}a_j=0,\\
              \lambda_{j,a_j}^t e_{g_{i,u}(a)}, & \textrm{otherwise},
		\end{array}\right.\notag\\
		&=&\left\{\begin{array}{ll}
		    \left(\lambda_{j,0}^t e_a^{(N/w)}+\sum\limits_{v=1}^{w-1}   (\lambda_{j,0}^t-\lambda_{j,v}^t) e_{a(j,v)}^{(N/w)}\right)\cdot \sum\limits_{b=0}^{N/w-1}(e_b^{(N/w)})^\top e_{g_{i,u}(b)}, & \textrm{if~}a_j=0,\\
              \lambda_{j,a_j}^t e_a^{(N/w)}\cdot \sum\limits_{b=0}^{N/w-1}(e_b^{(N/w)})^\top e_{g_{i,u}(b)}, & \textrm{otherwise},
		\end{array}\right.\notag\\
		&=&e_a^{(N/w)}\left(\sum\limits_{b=0}^{N/w-1}\lambda_{j,b_j}^t(e_b^{(N/w)})^\top e_b^{(N/w)}+\sum\limits_{b=0,b_j=0}^{N/w-1}\sum\limits_{v=1}^{w-1}(\lambda_{j,0}^t-\lambda_{j,v}^t)(e_b^{(N/w)})^\top e_{b(j,v)}^{(N/w)}\right)V_{i,u}\notag\\
		&=&e_a^{(N/w)}B_{t,j,i}\cdot V_{i,u}\notag\\
		&=&B_{t,j,i}[a,:]\cdot V_{i,u}
	\end{eqnarray}
where the second, third, and fourth equalities follow from \eqref{Eqn e_ae_bT}, \eqref{Eqn re Vu for liu}, and \eqref{Eqn the bmatrix bar{B}}, respectively. Applying \eqref{Eqn the repair and select matrix of code C0} to \eqref{Eqn the product between V[a] and Atj}, we have $S_{i,t}A_{t,j}=B_{t,j,i}R_{i,j}$, which finishes the proof of this case.	

\textit{Case 2.}   If $j\in [i+1: m)$, similar to the proof of Case 1, we also have  $S_{i,t}A_{t,j}=B_{t,j,i}R_{i,j}$.

\textit{Case 3.} If $j\in [m: m+i)$, then $(g_{i,u}(a))_{j-m}=a_{j-m}$ by \eqref{Eqn ga b}. By \eqref{eqn e_aQti}, we have
\begin{eqnarray}\label{Eqn the product between V[a] and Atj for case 3}
	V_{i,u}[a,:]\cdot A_{t,j}
		&=&\lambda_{j,a_{j-m}}^t e_{g_{i,u}(a)}\notag\\
		&=&\lambda_{j,a_{j-m}}^t e_a^{(N/w)}\cdot \sum\limits_{b=0}^{N/w-1}(e_b^{(N/w)})^\top e_{g_{i,u}(b)}\notag\\
		&=&e_a^{(N/w)}(\sum\limits_{b=0}^{N/w-1}\lambda_{j,b_{j-m}}^t(e_b^{(N/w)})^\top e_b^{(N/w)})V_{i,u}\notag\\
		&=&e_a^{(N/w)} B_{t,j,i}\cdot V_{i,u}\notag\\
		&=&B_{t,j,i}[a,:]\cdot V_{i,u}
	\end{eqnarray}
where the second, third, and fourth equalities follow from \eqref{Eqn e_ae_bT}, \eqref{Eqn re Vu for liu}, and \eqref{Eqn the bmatrix bar{B}}, respectively. Thus we have $S_{i,t}A_{t,j}=B_{t,j,i}R_{i,j}$ by combining \eqref{Eqn the repair and select matrix of code C0} and \eqref{Eqn the product between V[a] and Atj for case 3}.

\textit{Case 4.} If $j\in [m+i+1: n)$, similar to the proof of Case 3, we also have  $S_{i,t}A_{t,j}=B_{t,j,i}R_{i,j}$.

Collecting the above four cases, we can derive that ii)  holds for $0\le i\ne j<n$ with $i\not\equiv j\bmod m$. While
the proof of iii) is similar to the analysis in \eqref{Eqn the value of A(a,b)}; thus, we omit it.
\end{proof}

\end{document}